\documentclass[11pt]{article}
\usepackage{amssymb}
\usepackage{natbib}
\usepackage{mathrsfs}
\usepackage{amsmath,amssymb}
\usepackage{amsthm}
\usepackage{graphicx,color}
\usepackage{setspace}
\usepackage{footnote}
\usepackage{algorithm}

\usepackage[top=1in, bottom=1in, left=1in, right=1in]{geometry}
\usepackage{subfig}
\usepackage{soul}
\setstcolor{black}
\usepackage{multirow}
\usepackage{hhline}

\usepackage{multirow}

\renewcommand{\Pr}{\mathbb{P}}

\usepackage[noend]{algpseudocode}

\newtheorem{theorem}{Theorem}

\newtheorem{lemma}{Lemma}
\newtheorem{corollary}{Corollary}

\title{DECOrrelated feature space partitioning for distributed sparse regression}
\author{Xiangyu Wang \footnote{Department of Statistical Science, Duke University, USA.}
	\and
	David Dunson \footnotemark[1]
	\and
	Chenlei Leng \footnote{Department of Statistics, University of Warwick, UK.}}
\date{\today}
\begin{document}
\maketitle

\begin{abstract}
Fitting statistical models is computationally challenging when the sample size or the dimension of the dataset is huge. An attractive approach for down-scaling the problem size is to first partition the dataset into subsets and then fit using distributed algorithms. The dataset can be partitioned either horizontally (in the sample space) or vertically (in the feature space). While the majority of the literature focuses on sample space partitioning, feature space partitioning is more effective when $p\gg n$. Existing methods for partitioning features, however, are either vulnerable to high correlations or inefficient in reducing the model dimension. In this paper, we solve these problems through a new embarrassingly parallel framework named DECO for distributed variable selection and parameter estimation. In DECO, variables are first partitioned and allocated to $m$ distributed workers. The decorrelated subset data within each worker are then fitted via any algorithm designed for high-dimensional problems. We show that by incorporating the decorrelation step, DECO can achieve consistent variable selection and parameter estimation on each subset with (almost) no assumptions. In addition, the convergence rate is nearly minimax optimal for both sparse and weakly sparse models and does NOT depend on the partition number $m$. Extensive numerical experiments are provided to illustrate the performance of the new framework.
\end{abstract}

%=======================
\section{Introduction}
In modern science and technology applications, it has become routine to collect complex datasets with a huge number $p$ of variables and/or enormous sample
size $n$. Most of the emphasis in the literature has been on addressing large $n$ problems, with a common strategy relying on partitioning data samples into subsets and fitting a model containing all the variables to each subset \citep{mcdonald2009efficient,zhang2012communication,wang2013parallelizing,scott2013bayes,wang2015parallelizing,wang2014median,minsker2015geometric}. In scientific applications, it is much more common to have huge $p$ small $n$ data sets.  In such cases, a sensible strategy is to break the features into groups, fit a model separately to each group, and combine the results. We refer to this strategy as feature space partitioning, and to the large $n$ strategy as sample space partitioning.

\par
There are several recent attempts on parallel variable selection by partitioning the feature space. \citet{song2014split} proposed a Bayesian split-and-merge (SAM) approach in which variables are first partitioned into subsets and then screened over each subset. A variable selection procedure is then performed on the variables that survive for selecting the final model. One caveat for this approach is that the algorithm cannot guarantee the efficiency of screening, i.e., the screening step taken on each subset might select a large number of unimportant but correlated variables \citep{song2014split}, so the split-and-merge procedure could be ineffective in reducing the model dimension. Inspired by a group test, \citet{zhou2014parallel} proposed a parallel feature selection algorithm by repeatedly fitting partial models on a set of re-sampled features, and then aggregating the residuals to form scores for each feature. This approach is generic and efficient, but the performance relies on a strong condition that is almost equivalent to an independence assumption on the design.  

\par
Intuitively, feature space partitioning is much more challenging than sample space partitioning, mainly because of the correlations between features. A partition of the feature space would succeed only when the features across the partitioned subsets were mutually independent. Otherwise, it is highly likely that any model posed on the subsets is mis-specified and the results are biased regardless of the sample size. In reality, however, mutually independent groups of features may not exist; Even if they do, finding these groups is likely more challenging than fitting a high-dimensional model. Therefore, although conceptually attractive, feature space partitioning is extremely challenging.

% due to the following two reasons. The first is the correlations between variables. Unlike the sample space where independence is routinely assumed and preserved after partitioning, feature space often possesses complicated correlation structures. The following is a simple illustration of the correlation structure of the two spaces. It can be seen that a naive partition on the feature space might lead to huge information loss in the analysis. {\color{red}{Example?}} The second challenge is in the misspecified sub-models. Since each sub-model contains only an incomplete set of variables, there is a huge possibility that most of these models are "wrong", i.e., missing some true signals. Thus, any approach would indeed fail if applied directly without careful pretreatment.

On the other hand, feature space partitioning is straightforward if the features are independent. Motivated by this key fact, we propose a novel embarrassingly-parallel framework named DECO by \textit{decorrelating} the features before partitioning. With the aid of decorrelation, each subset of data after feature partitioning can now produce consistent estimates even though the model on each subset is intrinsically mis-specified due to missing features. To the best of our knowledge, DECO is the first embarrassingly parallel framework specifically designed to accommodate arbitrary correlation structure in the features. We show, quite surprisingly, that the DECO estimate, by leveraging the estimates from subsets, achieves the same convergence rate in $\ell_2$ norm and $\ell_\infty$ norm as the estimate obtained by using the full dataset, and that the rate does not depend on the number of partitions. In view of the huge computational gain and the easy implementation, DECO is extremely attractive for fitting large-$p$ data. %We also discuss the possibility of fitting large-$p$-large-$n$ data using DECO in the last section. 

\par
The most related work to DECO is \citet{jia2012preconditioning}, where a similar procedure was introduced to improve \emph{lasso}. Our work differs  substantially in various aspects. First, our motivation is to develop a parallel computing framework for fitting large-$p$ data by splitting features, which can potentially accommodate any penalized regression methods, while \citet{jia2012preconditioning} aim solely at complying with the irrepresentable condition for \emph{lasso}. Second, the conditions posed on the feature matrix are more flexible in DECO, and our theory, applicable for not only sparse signals but also those in $l_r$ balls, can be readily applied to the preconditioned \emph{lasso} in \citet{jia2012preconditioning}. 

\par
The rest of the paper is organized as follows. In Section 2, we detail the proposed framework. Section 3 provides the theory of DECO. In particular, we show that DECO is consistent for both sparse and weakly sparse models. Section 4 presents extensive simulation studies to illustrate the performance of our framework. In Section 5, we outline future challenges and future work. All the technical details are relegated to the Appendix. 

%=======================
\section{Motivation and the DECO framework}
Consider the linear regression model
\begin{equation} 
Y=X\beta+\varepsilon,\label{eq:lr}
\end{equation}
where $X$ is an $n\times p$ feature (design) matrix, $\varepsilon$  consists of $n$ \textit{i.i.d} random errors and $Y$ is the response vector.  A large class of approaches estimate $\beta$ by solving the following optimization problem
\begin{align*}
\hat\beta = arg\min_\beta \frac{1}{n}\|Y - X\beta\|_2^2 + 2\lambda_n\rho(\beta),
\end{align*}
where $\|\cdot\|_2$ is the $\ell_2$ norm and $\rho(\beta)$ is a penalty function. In this paper, we specialize our discussion to the $\ell_1$ penalty where $\rho(\beta)=\sum_{j=1}^p |\beta_j|$ \citep{tibshirani1996regression} to highlight the main
message of the paper. %General penalties will be discussed in the long version of this paper.
\par
As discussed in the introduction, a naive partition of the feature space will usually give unsatisfactory results under a parallel computing framework. That is why a decorrelation step is introduced. For data with $p\leq n$, the most intuitive way is to orthogonalize features via the singular value decomposition (SVD) of the design matrix as $X = UDV^T$, where $U$ is an $n\times p$ matrix, $D$ is an $p\times p$ diagonal matrix and $V$ an $p\times p$ orthogonal matrix. If we pre-multiply both sides of  \eqref{eq:lr} by $\sqrt{p}D^{-1}U^T$, we get
\begin{align}
\sqrt{p}D^{-1}U^TY = \sqrt{p}V^T\beta + \sqrt{p}D^{-1}U^T\varepsilon. \label{eq:T1}
\end{align}
It is obvious that the new features (the columns of $\sqrt{p}V^T$) are mutually orthogonal. An alternative approach to avoid doing SVD is to replace $\sqrt{p}D^{-1}U^T$ by $\sqrt{p}UD^{-1}U^T = (XX^T/p)^{\frac{+}{2}}$, where $A^+$ denotes the Moore-Penrose pseudo-inverse. Thus, we have
\begin{align}
(XX^T/p)^{\frac{+}{2}}Y = \sqrt{p}UV^T\beta + (XX^T/p)^{\frac{+}{2}}\varepsilon, \label{eq:T2}
\end{align}
where the new features (the columns of $\sqrt{p}UV^T$) are mutually orthogonal. Define the new data as $(\tilde Y, \tilde X)$. The mutually orthogonal property allows us to decompose $\tilde X$ column-wisely to $m$ subsets $\tilde X^{(i)}, i = 1, 2, \cdots, m$, and still retain consistency if one fits a linear regression on each subset. To see this, notice that each sub-model now takes a form of $\tilde Y = \tilde X^{(i)}\beta^{(i)} + \tilde W^{(i)}$ where $ W^{(i)} = \tilde X^{(-i)}\beta^{(-i)} + \tilde \varepsilon$ and $X^{(-i)}$ stands for variables not included in the $i^{th}$ subset. If, for example, we would like to compute the ordinary least squares estimates, it follows
\begin{align*}
\hat\beta^{(i)} &= (\tilde X^{(i)T}\tilde X^{(i)})^{-1}\tilde X^{(i)}\tilde Y\\
& = \beta^{(i)} + (\tilde X^{(i)T}\tilde X^{(i)})^{-1}\tilde X^{(i)}\tilde W^{(i)}\\
& = \beta^{(i)} + (\tilde X^{(i)T}\tilde X^{(i)})^{-1}\tilde X^{(i)}\tilde \varepsilon
\end{align*}
where we retrieve a consistent estimator that converges in rate as if the full dataset were used.
\par
When $p$ is larger than $n$, the new features are no longer exactly orthogonal to each other due to the high dimension. Nevertheless, the correlations between different columns are roughly of the order $\sqrt{\frac{\log p}{n}}$ for random designs, making the new features approximately orthogonal when $\log(p) \ll n$. This allows us to follow the same strategy of partitioning the feature space as in the low dimensional case. It is worth noting that when $p > n$, the SVD decomposition on $X$ induces a different form on the three matrices, i.e., $U$ is now an $n\times n$ orthogonal matrix, $D$ is an $n\times n$ diagonal matrix, $V$ is an $n\times p$ matrix, and $\big(\frac{XX^T}{p}\big)^{\frac{+}{2}}$ becomes $\big(\frac{XX^T}{p}\big)^{-\frac{1}{2}}$. 

In this paper, we primarily focus on datasets where $p$ is so large that a single computer is only able to store and perform operations on an $n\times q$ matrix ($n < q < p$) but not on an $n\times p$ matrix. %For large-$n$-large-$p$ problems, a promising approach that combines DECO with existing methods is discussed in the last section. 
Because the two decorrelation matrices yield almost the same properties, we will only present the algorithm and the theoretical analysis for $(XX^T/p)^{-1/2}$.

The concrete DECO framework consists of two main steps. Assume $X$ has been partitioned column-wisely into $m$ subsets $X^{(i)}, i=1, 2, \cdots, m$ (each with a maximum of $q$ columns) and distributed onto $m$ machines with $Y$. In the first stage, we obtain the decorrelation matrix $(XX^T/p)^{-1/2}$ or $\sqrt{p}D^{-1}U^T$ by computing $XX^T$ in a distributed way as $XX^T = \sum_{i = 1}^m X^{(i)}X^{(i)T}$ and perform the SVD decomposition on $XX^T$ on a central machine. In the second stage, each worker receives the decorrelation matrix, multiplies it to the local data $(Y, X^{(i)})$ to obtain $(\tilde Y, \tilde X^{(i)})$, and fits a penalized regression. When the model is assumed to be exactly sparse, we can potentially apply a refinement step by re-estimating coefficients on all the selected variables simultaneously on the master machine via ridge regression. The details are provided in Algorithm \ref{alg:1}.
%===================================
\begin{algorithm}[!ht]
  \caption{{\em The DECO framework}}
  \label{alg:1}
  \begin{algorithmic}[1]
    \Require
    \State Input $(Y, X), p, n, m, \lambda_n$. Standardize $X$ and $Y$ to $x$ and $y$ with mean zero;
    \State Partition (arbitrarily) $(y, x)$ into $m$ disjoint subsets $(y, x^{(i)})$ and distribute to $m$ machines;
    \Ensure \textbf{1 : Decorrelation} 
    \State $F = 0$ initialized on the master machine;
    \For{$i = 1$ {\bf to} $m$ }
    \State $F = F + x^{(i)}x^{(i)T}$;
    \EndFor
    \State $\bar F = \sqrt{p}\big(F + r_1I_p\big)^{-1/2}$ on the master machine and then pass back; \label{line:1}
    \For{$i = 1$ {\bf to} $m$ }
    \State $\tilde y = \bar Fy$ and $\tilde x^{(i)} = \bar Fx^{(i)}$;
    \EndFor
    \Ensure \textbf{2 : Estimation}
    \For{$i = 1$ {\bf to} $m$ }
    \State $\hat\beta^{(i)} = arg\min_\beta \frac{1}{n}\|\tilde y - \tilde x^{(i)}\beta\|_2^2 + 2\lambda_n\rho(\beta)$;
    \EndFor
    \State $\hat\beta = (\hat\beta^{(1)}, \hat\beta^{(2)},\cdots,\hat\beta^{(m)})$ on the master machine;
    \State $\hat\beta_0 = mean(Y) - mean(X)^T\hat\beta$ for intercept.
    \Ensure \textbf{3 : Refinement (optional)}
    \If{$\#\{\hat\beta \neq 0\} \geq n$} \label{line:2}
    \State \Comment Sparsification is needed before ridge regression.
    \State $\mathcal{M} = \{k~: |\hat\beta_k|\neq 0\}$;
    \State $\hat\beta_\mathcal{M} = arg\min_\beta \frac{1}{n}\|\tilde y - \tilde x_\mathcal{M}\beta\|_2^2 + 2\lambda_n\rho(\beta)$;
    \EndIf \label{line:3}
    \State $\mathcal{M} = \{k~: |\hat\beta_k|\neq 0\}$;
    \State $\hat\beta_\mathcal{M} = (X_\mathcal{M}^TX_\mathcal{M} + r_2I_{|\mathcal{M}|})^{-1}X_\mathcal{M}^TY$;
    \Return $\hat\beta$;
  \end{algorithmic}
\end{algorithm}

Line \ref{line:2} - \ref{line:3} in Algorithm \ref{alg:1} are added only for the data analysis in Section 5.3, in which $p$ is so massive that $\log(p)$ would be comparable to $n$.  For such extreme cases, the algorithm may not scale down the size of $p$ sufficiently for even obtaining a ridge regression estimator afterwards. Thus,  a further sparsification step is recommended. This differs fundamentally from the merging step in SAM \citep{song2014split} in that DECO needs this step only for extreme cases where $\log (p)\sim n$, while SAM always requires a merging step regardless of the relationship between $n$ and $p$. The condition in Line 16 is barely triggered in our experiments (only in Section 5.3), but is crucial for improving the performance for extreme cases. In Line \ref{line:1}, the algorithm inverts $XX^T + r_1I$ instead of $XX^T$ for robustness, because the rank of $XX^T$ after standardization will be $n - 1$. Using ridge refinement instead of ordinary least squares is also for robustness. The precise choice of $r_1$ and $r_2$ will be discussed in the numerical section. 

Penalized regression fitted using regularization path usually involves a computational complexity of $\mathcal{O}(knp + kd^2)$, where $k$ is the number of path segmentations and $d$ is the number of features selected. Although the segmentation number $k$ could be as bad as $(3^p + 1)/2$ in the worst case \citep{mairal2012complexity}, real data experience suggests that $k$ is on average $\mathcal{O}(n)$ \citep{rosset2007piecewise}, thus the complexity for DECO takes a form of $\mathcal{O}\big(n^3 + n^2\frac{p}{m} + m\big)$ in contrast to the full \emph{lasso} which takes a form of $\mathcal{O}(n^2p)$. %The $\mathcal{O}(n^3)$ complexity comes from computing the inverse square root of $X^TX/n$, which can be further brought down to $\mathcal{O}(n^{2.376})$ if Coppersmith–Winograd algorithm and Denman-Beavers iteration method are used with some good initialization.
As $n$ is assumed to be small, using DECO can substantially reduce the computational cost if $m$ is properly chosen.
%===================================
\section{Theory}
In this section, we provide theoretical justification for DECO on random feature matrices. We specialize our attention to \emph{lasso} due to page limits and will provide the theory on general penalties in the long version. Because the two decorrelation matrices yield similar consistency properties, the theory will be stated only for $(XX^T/p)^{-\frac{1}{2}}$. This section consists of two parts. The first part provides preliminary results for \emph{lasso}, when strong conditions on the feature matrix are imposed. In the second part, we adapt these results to DECO and show that the decorrelated data will automatically satisfy the 
conditions on the feature matrix even when the original features are highly correlated.

\subsection{Deterministic results for \emph{lasso} with conditions on the feature matrix}
Define $Q = \{1,2,\cdots, p\}$ and let $A^c$ be $Q\setminus A$ for any set $A\subseteq Q$. The following theorem provides deterministic conditions for \emph{lasso} on sup-norm convergence, $\ell_2$-norm  convergence and sign consistency.
\begin{theorem}\label{thm:1}
Under model \eqref{eq:lr}, denote the solution to the \emph{lasso} problem as
\begin{align*}
\hat\beta = arg\min_{\beta\in \mathcal{R}^p } \frac{1}{n}\big\|Y - X\beta\big\|_2^2 + 2\lambda_n\|\beta\|_1.
\end{align*}
Define $W = Y - X\beta_*$, where $\beta_*$ is the true value of $\beta$.
For any arbitrary subset $J\subseteq Q$ ($J$ could be $\emptyset$), if $X$ satisfies that
\begin{enumerate}
\item $M_1 \leq |x_i^Tx_i/n|  \leq M_2$, for some $0 < M_1 < M_2$ and all $i$,
\item $\max_{i\neq j} |x_i^Tx_j/n| \leq \min\bigg\{\frac{1}{\gamma_1 s},~ \gamma_2\lambda_n^{q}\bigg\}$, for $\gamma_1>\frac{32}{M_1}, \gamma_2 \geq 0$, $q\geq 0$ and $s = |J|$,
\item $\|\frac{1}{n}X^TW\|_\infty \leq \lambda_n/2$,
\end{enumerate}
then any solution to the \emph{lasso} problem satisfies that
\begin{align*}
\|\hat\beta - \beta_*\|_\infty \leq &\frac{3M_1\gamma_1 + 51}{2M_1(M_1\gamma_1 - 7)}\lambda_n + \frac{4M_1\gamma_1\gamma_2 + 36\gamma_2}{M_1(M_1\gamma_1 - 7)}\|\beta_{*J^c}\|_1\lambda_n^{q} + \frac{8\sqrt{3\gamma_2}}{M_1\sqrt{M_1\gamma_1 - 7}}\|\beta_{*J^c}\|_1^{\frac{1}{2}}\lambda_n^{\frac{1 + q}{2}},
\end{align*}
where $\beta_{*J^c}$ is the sub-vector of $\beta_*$ consisting of coordinates in $J^c$ and
\begin{align*}
\|\hat\beta - \beta_*\|_2^2 \leq &\frac{18\gamma_1^2 s\lambda_n^2}{(M_1\gamma_1 - 32)^2} + 6\lambda_n\|\beta_{*J^c}\|_1 + 32\gamma_2\lambda_n^q\|\beta_{*J^c}\|_1^2.
\end{align*}
Furthermore, if $\beta_{*J^c} = 0$ and $\min_{k\in J}|\beta_{*k}| \geq \frac{2}{M_1}\lambda_n$, then the solution is unique and sign consistent, that is,
\begin{align*}
sign(\hat\beta_k) = sign(\beta_{*k}),~\forall k\in J\quad\mbox{and}\quad \hat\beta_k = 0, ~\forall k\in J^c.
\end{align*}
\end{theorem}
Theorem \ref{thm:1} partly extends the results in \citet{bickel2009simultaneous} and \citet{lounici2008sup}. The proof is provided in Appendix A. Theorem \ref{thm:1} can lead to some useful results. In particular, we investigate two types of models when $\beta_*$ is either exactly sparse or in an $l_r$-ball defined as $\mathbb{B}(r, R) = \{v\in\mathcal{R}^p:~ \sum_{k = 1}^p |v_k|^r \leq R\}$. For the exactly sparse model, we have the following result.
\begin{corollary}[s-sparse]\label{co:1}
Assume that $\beta_*\in \mathcal{R}^p$ is an s-sparse vector with $J$ containing all non-zero indices. If Condition 1 and 3 in Theorem \ref{thm:1} hold and $\max_{i\neq j} |x_i^Tx_j/n|\leq \frac{1}{\gamma_1 s}$ for some $\gamma_1>32/M_1$, then we have
\begin{align*}
\|\hat\beta - \beta_*\|_\infty &\leq \frac{3M_1\gamma_1 + 51}{2M_1(M_1\gamma_1 - 7)}\lambda_n,\\
\|\hat\beta - \beta_*\|_2^2 &\leq \frac{18\gamma_1^2 s\lambda_n^2}{(M_1\gamma_1 - 32)^2}.
\end{align*}
Further, if $\min_{k\in J}|\beta_k| \geq \frac{2}{M_1}\lambda_n$, then $\hat\beta$ is sign consistent.
\end{corollary}
The sup-norm convergence in Corollary \ref{co:1} resembles the results in \citet{lounici2008sup}. For the $l_r$-ball we have
\begin{corollary}[$l_r-$ ball]\label{co:2}
Assume $\beta_*\in \mathbb{B}(r, R)$. If condition 1 and 3 in Theorem \ref{thm:1} hold and $\max_{i\neq j} |x_i^Tx_j/n|\leq \frac{\lambda_n^r}{\gamma_1 R}$ for some $\gamma_1 > 32/M_1$, then we have
\begin{align*}
\|\hat\beta - \beta_*\|_\infty \leq &\bigg(\frac{3M_1\gamma_1 + 51}{2M_1(M_1\gamma_1 - 7)} + \frac{4M_1\gamma_1 + 36}{M_1(M_1\gamma_1 - 7)} + \frac{8\sqrt{3}}{M_1\sqrt{M_1\gamma_1 - 7}}\bigg)\lambda_n,\\
\|\hat\beta - \beta_*\|_2^2\leq &\bigg(\frac{18\gamma_1^2}{(M_1\gamma_1 - 32)^2} + 38\bigg)R\lambda_n^{2 - r}.
\end{align*}
\end{corollary}

\subsection{Results for DECO without conditions on the feature matrix}
In this part, we apply the previous results from the \emph{lasso} to DECO, but without the restrictive conditions on the feature matrix. In particular, we prove the consistency results for the estimator obtained after Stage 2 of DECO, while the consistency of Stage 3 will then follow immediately. Recall that DECO works with the decorrelated data $\tilde X$ and $\tilde Y$, which are distributed on $m$ different machines. Therefore, it suffices for us to verify the conditions needed by \emph{lasso} for all pairs $(\tilde Y, \tilde X^{(i)}), i = 1,2,\cdots,m$. For simplicity, we assume that $\varepsilon$ follows a sub-Gaussian distribution and $X\sim N(0, \Sigma)$ throughout this section, although the theory can be easily extended to the situation where $X$ follows an elliptical distribution and $\varepsilon$ is heavy-tailed. %These results will appear in the long version of this paper.
\par
As described in the last section, DECO fits the following linear regression on each worker
\begin{align*}
\tilde Y = \tilde X^{(i)}\beta^{(i)} + \tilde W^{(i)},\quad\mbox{and}\quad
\tilde W^{(i)} = \tilde X^{(-i)}\beta^{(-i)} + \tilde \varepsilon,
\end{align*}
where $X^{(-i)}$ stands for variables not included in the $i^{th}$ subset.
To verify Condition 1 and Condition 2 in Theorem \ref{thm:1}, we cite a result from \citet{wang2015consistency} (Lemma \ref{lemma:1} in Appendix C) which proves the boundedness of $M_1$ and $M_2$ and that $\max_{i\neq j}|\tilde x_i^T\tilde x_j^T|/n$ is small. Verifying Condition 3 is the key to the whole proof. Different from the conventional setting, $\tilde W$ now contains non-zero signals that are not independent from the predictors. This requires us to accurately capture the behavior of the following two terms
$\max_{k\in Q}\big|\frac{1}{n}\tilde x_k^T \tilde X^{(-k)}\beta_{*}^{(-k)}\big| $ and $ \max_{k\in Q} \big|\frac{1}{n}\tilde x_k^T\tilde \varepsilon\big|$,
for which we have
\begin{lemma}\label{lemma:2}
Assume that $\varepsilon$ is a sub-Gaussian variable with a $\psi_2$ norm of $\sigma$ and $X\sim N(0, \Sigma)$. Define $\sigma_0^2 = var(Y)$. If $p > c_0 n$ for some $c_0 > 1$, then we have for any $t > 0$
\begin{align*}
&\Pr\bigg(\max_{k\in Q} \frac{1}{n}|\tilde x_k^T\tilde \varepsilon| > \frac{\sigma t}{\sqrt{n}}\bigg) \leq 2p\exp\bigg(-\frac{c_*c_0^2}{2c^*c_2(1 - c_0)^2} t^2\bigg) + 4pe^{-Cn},\\
&\Pr\bigg(\max_{k\in Q}\frac{1}{n}\big|\tilde x_k^T \tilde X^{(-k)}\beta_{*}^{(-k)}\big| \geq \frac{\sqrt{\sigma_0 - \sigma}t}{\sqrt{n}}\bigg)\leq 2p\exp\bigg(-\frac{c_*^3}{2c_4^2c^{*2}}t^2\bigg) + 5pe^{-Cn},
\end{align*}
where $C, c_1, c_2, c_4, c_*, c^*$ are defined in Lemma \ref{lemma:1}.
\end{lemma}
We now provide the main results of the paper.
\begin{theorem}[s-sparse]\label{thm:2}
Assume that $\beta_*$ is an s-sparse vector. Define $\sigma_0^2 = var(Y)$. 
For any $A>0$ we choose $\lambda_n = A\sigma_0\sqrt{\frac{\log p}{n}}$. Now if $p > c_0n$ for some $c_0 > 1$ and $64C_0^2A^2s^2\frac{\log p}{n} \leq 1$, then with probability at least $1 - 8p^{1 - C_1A^2} - 18p e^{-Cn}$ we have
\begin{align*}
\|\hat\beta - \beta_*\|_\infty &\leq \frac{5C_0A\sigma_0}{8}\sqrt{\frac{\log p}{n}},\\
 \|\hat\beta - \beta_*\|_2^2 &\leq \frac{9C_0A^2\sigma_0^2}{8} \frac{s\log p}{n},
\end{align*}
where $C_0 = \frac{8c^*}{c_1c_*}$ and $C_1 = \min\{\frac{c_*c_0^2}{8c^*c_2(1 - c_0)^2}, \frac{c_*^3}{8c_4^2c^{*2}}\}$ are two constants and $c_1,c_2,c_4,c_*, c^*, C$ are defined in Lemma \ref{lemma:1}. Furthermore, if we have 
\[
\min_{\beta_k\neq 0}|\beta_k|\geq \frac{C_0A\sigma_0}{4}\sqrt{\frac{\log p}{n}},
\]
then $\hat\beta$ is sign consistent.
\end{theorem}
Theorem \ref{thm:2} looks a bit surprising since the convergence rate does not depend on $m$. This is mainly because the bounds used to verify the three conditions in Theorem \ref{thm:1} hold uniformly on all subsets of variables. For subsets where no true signals are allocated, \emph{lasso} will estimate all coefficients to be zero under suitable choice of $\lambda_n$, so that the loss on these subsets will be exactly zero. Thus, when summing over all subsets, we retrieve the $\frac{s\log p}{n}$ rate. In addition, it is worth noting that Theorem \ref{thm:2} guarantees the $\ell_\infty$ convergence and sign consistency for \emph{lasso} without assuming the irrepresentable condition \citep{zhao2006model}. A similar but weaker result was obtained in \citet{jia2012preconditioning}.

\begin{theorem}[$l_r$-ball]\label{thm:3}
Assume that $\beta_*\in\mathbb{B}(r, R)$ and all conditions in Theorem \ref{thm:2} except that $64C_0^2A^2s^2\frac{\log p}{n} \leq 1$ are now replaced by $64C_0^2A^2R^2\big(\frac{\log p}{n}\big)^{1 - r} \leq 1$. Then with probability at least $1 - 8p^{1 - C_1A^2} - 18p e^{-Cn}$, we have
\begin{align*}
\|\hat\beta - \beta_*\|_\infty &\leq \frac{3C_0A\sigma_0}{2}\sqrt{\frac{\log p}{n}},\\
\|\hat\beta - \beta_*\|_2^2 &\leq \bigg(\frac{9C_0}{8} + 38\bigg) (A\sigma_0)^{2 - r}R\bigg(\frac{\log p}{n}\bigg)^{1 - \frac{r}{2}}.
\end{align*}
\end{theorem}

Note that $\sigma_0^2 = var(Y)$ instead of $\sigma$ appears in the convergence rate in both Theorem \ref{thm:2} and \ref{thm:3}, which is inevitable due to the nonzero signals contained in $\tilde W$. Compared to the estimation risk using full data, the results in Theorem \ref{thm:2} and \ref{thm:3} are similar up to a factor of $\sigma^2/\sigma_0^2 = 1 - \hat R^2$, where $\hat R^2$ is the coefficient of determination. Thus, for a model with an $\hat R^2 = 0.8$, the risk of DECO is upper bounded by five times the risk of the full data inference. The rates in Theorem \ref{thm:2} and \ref{thm:3} are nearly minimax-optimal \citep{ye2010rate,raskutti2009minimax}, but the sample requirement $n \asymp s^2$ is slightly off the optimal. This requirement is rooted in the $\ell_\infty$-convergence and sign consistency and is almost unimprovable for random designs. We will detail this argument in the long version of the paper.

\section{Experiments}
In this section, we present the empirical performance of DECO via extensive numerical experiments. In particular, we compare DECO after 2 stage fitting ({\bf DECO-2}) and DECO after 3 stage fitting ({\bf DECO-3}) with the full data \emph{lasso} ({\bf lasso-full}), the full data \emph{lasso} with ridge refinement ({\bf lasso-refine}) and \emph{lasso} with a naive feature partition without decorrelation ({\bf lasso-naive}). This section consists of three parts. In the first part, we run {\bf DECO-2} on some simulated data and monitor its performance on one randomly chosen subset that contains part of the true signals. In the second part, we verify our claim in Theorem \ref{thm:2} and \ref{thm:3} that the accuracy of DECO does not depend on the subset number. In the last part, we provide a comprehensive evaluation of DECO's performance by comparing DECO with other methods under various correlation structures. 
\par
The synthetic datasets are from model \eqref{eq:lr} with $X\sim N(0, \Sigma)$ and $\varepsilon\sim N(0, \sigma^2)$. The variance $\sigma^2$ is chosen such that $\hat R^2 = var(X\beta)/var(Y)=0.9$. For evaluation purposes, we consider five different structures of $\Sigma$ as below.

\par
\noindent \textbf{Model (i) \emph{Independent predictors}}.
The support of $\beta$ is $S=\{1,2,3,4,5\}$. We generate $X_i$ from a standard multivariate
normal distribution with independent components. The coefficients are specified as
\[
\beta_i = \left\{
\begin{array}{ll}
(-1)^{Ber(0.5)}\bigg(|N(0,1)|+5\sqrt{\frac{\log p}{n}}\bigg) & i\in S\\
0 & i\not\in S.
\end{array}
\right.
\]

\noindent \textbf{Model (ii) \emph{Compound symmetry}}. 
All predictors are equally correlated with correlation
$\rho = 0.6$. The coefficients are the same as those in Model (i).

\noindent \textbf{Model (iii) \emph{Group structure}}. 
This example is Example 4 in \citet{zou2005regularization}, for which we
allocate the 15 true variables into three groups. Specifically, the predictors are generated as
$x_{1+3m} = z_1+N(0, 0.01)$, $x_{2+3m} = z_2+N(0, 0.01)$ and $x_{3+3m} = z_3+N(0, 0.01)$,
where $m=0,1,2,3,4$ and $z_i\sim N(0,1)$ are independent. The coefficients are set as
$ \beta_i = 3, ~i=1,2,\cdots,15;~ \beta_i = 0, ~i=16,\cdots,p$.

\noindent \textbf{Model (iv) \emph{Factor models}}. 
This model is considered in \citet{meinshausen2010stability}. Let $\phi_j,
j=1,2,\cdots,k$ be independent standard normal variables. We set
predictors as $x_i = \sum_{j=1}^k \phi_j f_{ij}+\eta_i$, where $f_{ij}$ and
$\eta_i$ are independent standard normal random variables. The
number of factors is chosen as $k=5$ in the simulation while
the coefficients are specified the same as in Model (i).

\noindent \textbf{Model (v) \emph{$\ell_1$-ball}}.
This model takes the same correlation structure as Model (ii), with the coefficients
drawn from Dirichlet distribution
$\beta \sim Dir\big(\frac{1}{p}, \frac{1}{p}, \cdots, \frac{1}{p}\big)\times 10$.
This model is to test the performance under a weakly sparse assumption on $\beta$, since $\beta$
is non-sparse satisfying $\|\beta\|_1 = 10$.
\par
Throughout this section, the performance of all the methods is evaluated in terms of four metrics: the number of false positives ({\bf \# FPs}), 
the number of false negatives ({\bf \# FNs}), the mean squared error $\|\hat\beta - \beta_*\|_2^2$ ({\bf MSE}) and 
the computational time ({\bf runtime}). We use \texttt{glmnet} \citep{friedman2010regularization} to fit \emph{lasso} and choose the tuning parameter 
using the extended BIC criterion \citep{chen2008extended} with $\gamma$ fixed at $0.5$. For DECO, the features are partitioned randomly in Stage 1
and the tuning parameter $r_1$ is fixed at $1$ for {\bf DECO-3}. Since {\bf DECO-2} does not involve any refinement step, we choose $r_1$ to be $10$
to aid robustness. The ridge parameter $r_2$ 
is chosen by 5-fold cross-validation for both {\bf DECO-3} and {\bf lasso-refine}. All the algorithms 
are coded and timed in \emph{Matlab} on computers with Intel i7-3770k cores. For any embarrassingly parallel algorithm, we
report the preprocessing time plus the longest runtime of a single machine as its runtime.

\subsection{Monitor DECO on one subset}
In this part, using data generated from Model (ii), we illustrate the performance of DECO on one randomly chosen subset after partitioning. 
The particular subset we examine contains two nonzero coefficients $\beta_1$ and $\beta_2$ 
with $98$ coefficients, randomly chosen, being zero. We either fix $p = 10,000$ and change $n$ from $100$ to $500$, or 
fix $n$ at $500$ and change $p$ from $2,000$ to $10,000$ to simulate datasets. 
We fit {\bf DECO-2}, {\bf lasso-full} and {\bf lasso-naive} to $100$ simulated datasets, and monitor their performance on that particular subset.
The results are shown in Fig \ref{fig:1} and \ref{fig:2}.
\begin{figure}[!tb]
\centering
\includegraphics[height = 12cm]{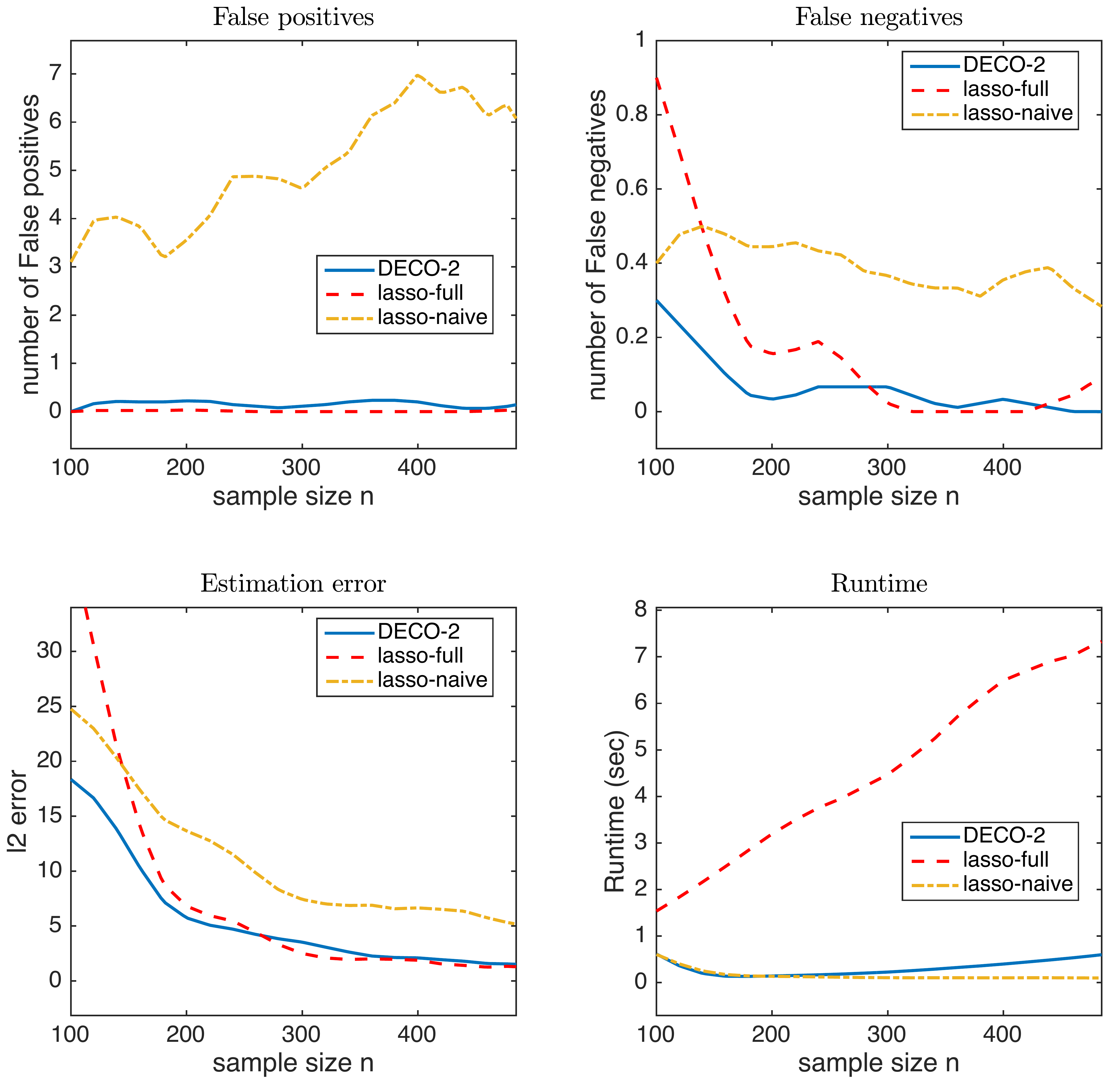}
\caption{Performance of DECO on one subset with $p = 10,000$ and different $n's$.}
\label{fig:1}
\end{figure}

\begin{figure}[!tb]
\centering
\includegraphics[height = 12cm]{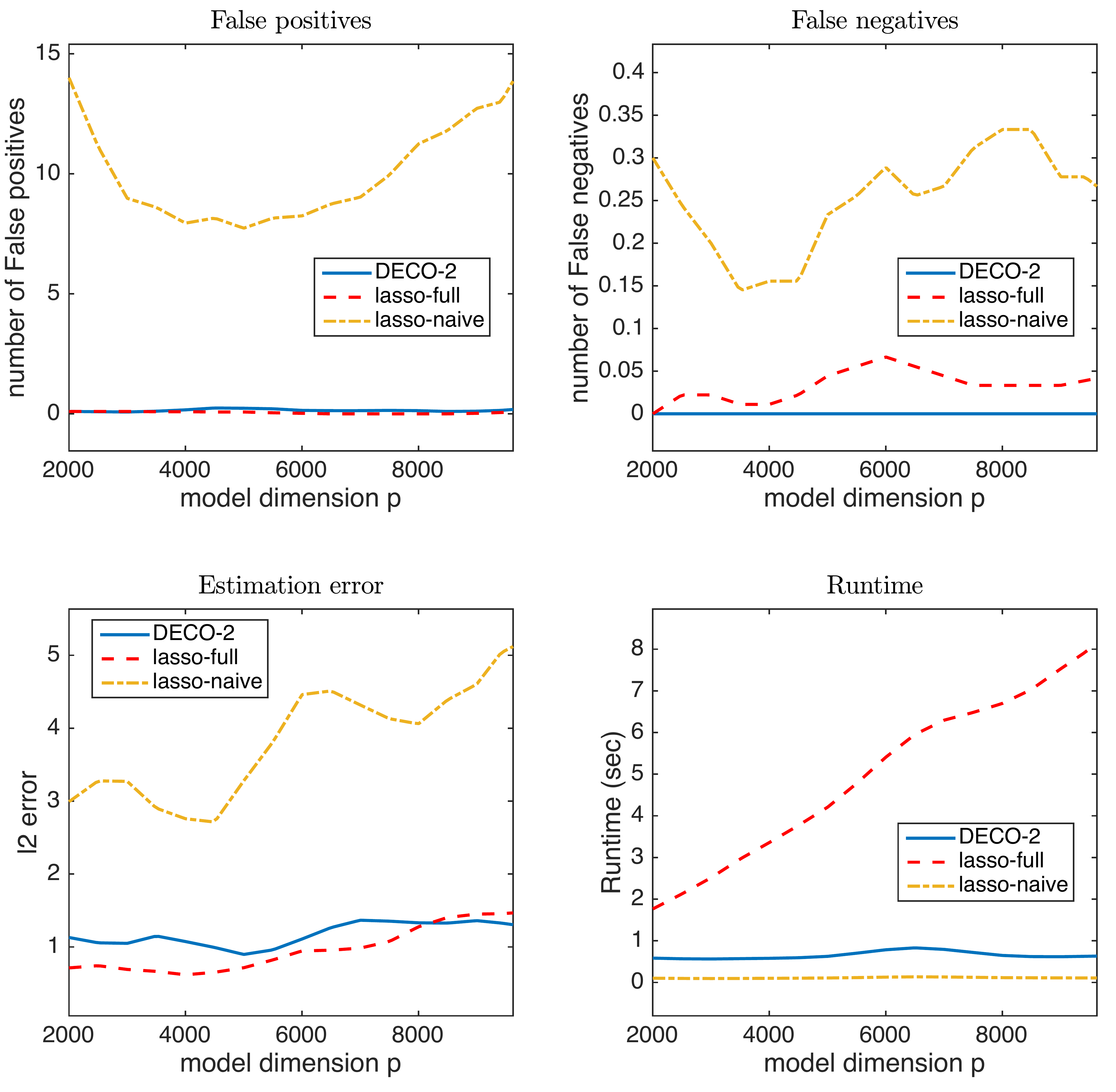}
\caption{Performance of DECO on one subset with $n = 500$ and different $p's$.}
\label{fig:2}
\end{figure}

It can be seen that, though the sub-model on each subset is mis-specified,  DECO performs as if the full dataset were
used as its performance is on par with {\bf lasso-full}. On the other hand,  {\bf lasso-naive} fails completely. This result clearly highlights the 
advantage of decorrelation before feature partitioning.

\subsection{Impact of the subset number $m$}
As shown in Theorem \ref{thm:2} and \ref{thm:3}, the performance of DECO does not depend on the number
of partitions $m$. We verify this property by using Model (ii) again. This time, we fix $p = 10,000$ and $n = 500$,
and vary $m$ from $1$ to $200$. We compare the performance of {\bf DECO-2} and {\bf DECO-3} with {\bf lasso-full}
and {\bf lasso-refine}. The averaged results from $100$ simulated datasets are plotted in Fig \ref{fig:3}. 
\begin{figure}[!tb]
\centering
\includegraphics[height = 12cm]{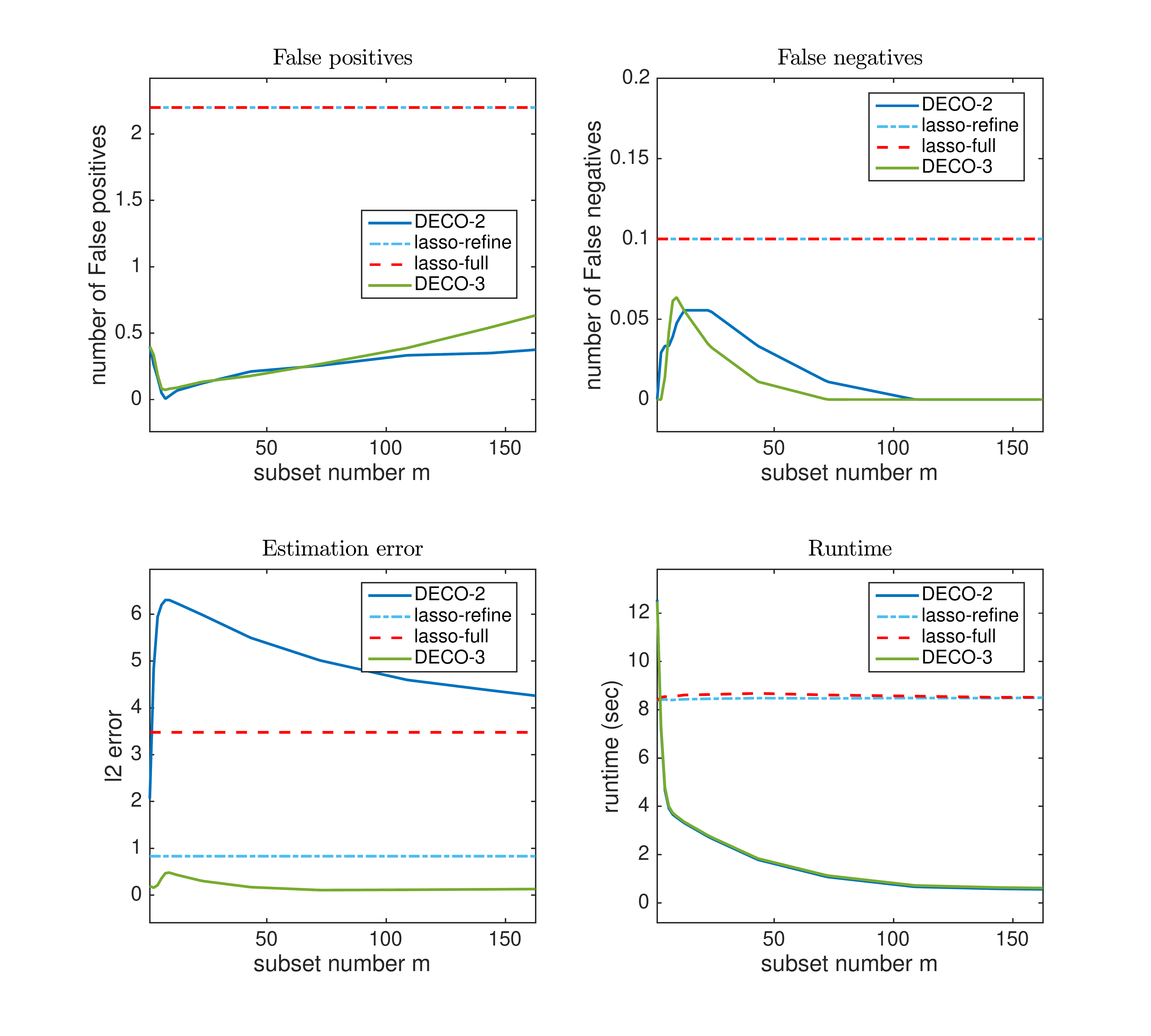}
\caption{Performance of DECO with different number of subsets.}
\label{fig:3}
\end{figure}
Since $p$ and $n$ are both fixed, {\bf lasso-full} and {\bf lasso-refine} are expected to perform stably over different $m'$s.  {\bf DECO-2} and {\bf DECO-3} also maintain a stable performance regardless of the value of $m$. In addition, {\bf DECO-3}
achieves a similar performance to and sometimes better accuracy than {\bf lasso-refine}, possibly because the irrepresentable condition is satisfied 
after decorrelation (See the discussions after Theorem \ref{thm:2}).

\subsection{Comprehensive comparison}
In this section, we compare all the methods under the five different correlation structures. The model dimension and the sample size are
fixed at $p = 10,000$ and $n = 500$ respectively and the number of subsets is fixed as $m = 100$. For each model, we simulate
$100$ synthetic datasets and record the average performance in Table \ref{tab:1}

\begin{table}[!htbp]
\caption{Results for five models with $(n, p) = (500, 10000)$}
%\vskip 0.15in
\label{tab:1}
\begin{center}
\begin{tabular}{l|l|ccccc}
\hline
Model      &          & {\bf DECO-3} & {\bf DECO-2} & {\bf lasso-refine} & {\bf lasso-full} & {\bf lasso-naive}\\
\hline\hline
           & MSE      &  0.102    & 3.502     & 0.104       & 0.924        & 3.667\\
(i)  & \# FPs   &  0.470    & 0.570     & 0.420       & 0.420        & 0.650\\
           & \# FNs   &  0.010    & 0.020     & 0.000       & 0.000        & 0.010\\
           & Time  &  65.5     & 60.3      & 804.5        & 802.5         & 9.0\\
\hline
           & MSE      &  0.241    & 4.636     & 1.873       & 3.808        & 171.05\\
(ii) & \# FPs   &  0.460    & 0.550     & 2.39       & 2.39        & 1507.2\\
           & \# FNs   &  0.010    & 0.030     & 0.160       & 0.160        & 1.290\\
           & Time  &  66.9     & 61.8      & 809.2        & 806.3         & 13.1\\
\hline
           & MSE     &  6.620    & 1220.5     & 57.74       & 105.99       & 1235.2\\
(iii)& \# FPs   &  0.410    & 0.570     & 0.110       & 0.110        & 1.180 \\
           & \# FNs   &  0.130    & 0.120     & 3.93       & 3.93       & 0.110\\
           & Time  &  65.5     & 60.0      & 835.3        & 839.9         & 9.1\\
\hline
           & MSE     &  0.787    & 5.648     & 11.15       & 6.610        & 569.56\\
(iv)&  \# FPs   &  0.460    & 0.410     & 19.90       & 19.90        & 1129.9\\
           & \# FNs   &  0.090    & 0.100     & 0.530       & 0.530       & 1.040\\
           & Time  &  69.4     & 64.1      & 875.1        & 880.0         & 14.6\\
\hline
(v)  & MSE     &  ---    & 2.341    & ---       & 1.661        & 356.57\\
           & Time  &  ---     & 57.5     & ---        & 829.5         & 13.3\\
\hline
\hline
\end{tabular}
\end{center}
\end{table}

% \newcommand{\colw}{0.4cm}
% \begin{table*}[!tb]
% \caption{Results for $(n, p) = (500, 10,000)$}
% \label{tab:1}
% \begin{center}
% \begin{tabular}{l|p{\colw}p{\colw}p{\colw}p{\colw}|p{\colw}p{\colw}p{\colw}p{\colw}p{\colw}p{\colw}p{\colw}p{\colw}p{\colw}p{\colw}p{\colw}p{\colw}p{\colw}p{\colw}}
% \hline
% Models &  & (i) &  &  &  & (ii) &  &  &  & (iii) &  &  &  & (iv) &  &  & (v) &   \\
% \hline
%        & MSE & \# FPs & \# FNs & Time & MSE &  \# FPs &  \# FNs & Time & MSE &  \# FPs & \# FNs & Time & MSE & \# FPs & \# FNs & Time & MSE & Time  \\
% \hline\hline
% {\bf DECO-3}       & 0.102 & 0.470 & 0.010 & 65.5  & 0.241 & 0.460 & 0.010 & 66.9  & 6.620  & 0.410 & 0.130 & 65.5  & 0.787 & 0.460 & 0.090 & 69.4 & ---  & ---  \\
% {\bf DECO-2}       & 3.502 & 0.570 & 0.020 & 60.3  & 4.636 & 0.550 & 0.030 & 61.8  & 1220.5 & 0.570 & 0.120 & 60.0  & 5.648 & 0.410 & 0.100 & 64.1 & 2.341 & 57.5  \\
% {\bf lasso-refine} & 0.104 & 0.420 & 0.000 & 804.5 & 1.873 & 2.390 & 0.160 & 809.2 & 57.74  & 0.110 & 3.930 & 835.3 & 11.15 & 19.90 & 0.530 & 875.1& --- & ---  \\
% {\bf lasso-full}   & 0.924 & 0.420 & 0.000 & 802.5 & 3.808 & 2.390 & 0.160 & 806.3 & 105.99 & 0.110 & 3.930 & 839.9 & 6.610 & 19.90 & 0.530 & 880.0& 1.661 & 829.5  \\
% {\bf lasso-naive}  & 3.667 & 0.650 & 0.010 & 9.0 & 171.05  & 1507.2 & 1.290 & 13.1 & 1235.2 & 1.180 & 0.110 & 9.1 & 569.56  & 1129.9& 1.040 & 14.6 & 356.57 & 13.3  \\
% \hline
% \hline
% \end{tabular}
% \end{center}
% \end{table*}

Several conclusions can be drawn from Table \ref{tab:1}. First, when all variables are independent as in Model (i), 
{\bf lasso-naive} performs similarly to {\bf DECO-2} because no decorrelation is needed in this simple case. 
However, {\bf lasso-naive} fails completely for the other four models when correlations are presented. Second,
{\bf DECO-3} achieves the overall best performance. The better estimation error over {\bf lasso-refine} is due
to the better variable selection performance, since the irrepresentable condition is not needed for DECO. Finally,
{\bf DECO-2} performs similarly to {\bf lasso-full} and the difference is as expected according to the discussions after Theorem \ref{thm:3}.
%i.e., 
%for a model with $\hat R^2 = 0.9$, the error of {\bf DECO-2} is upper bounded by $1/(1 - 0.9) = 10$ times the error of
%{\bf lasso-full}.

\section{Real data}
We illustrate the competitve performance of DECO via three real datasets that cover a range of high
dimensionalities, by comparing {\bf DECO-3} to {\bf lasso-full}, {\bf lasso-refine} and {\bf lasso-naive}
in terms of prediction error and computational time. The algorithms are configured in the same way as in Section 4. Although DECO allows arbitrary partitioning
(not necessarily random) over the feature space, for simplicity, we confine our attention to random partitioning.
In addition, we perform {\bf DECO-3} multiple times on the same dataset to ameliorate the uncertainty due to the randomness in partitioning.

\subsection{Student performance dataset} 
We look at one of the two datasets used for evaluating student achievement in two Portuguese schools \citep{cortez2008using}. The data attributes include school related features that were collected by using school reports and questionnaires. The particular dataset used here provides the students' performance in mathematics. The goal of the research is to predict the final grade (range from 0 to 20). The original data set contains 395 students and 32 raw attributes. The raw attributes are recoded as 40 attributes and form 767 features after adding interaction terms. To reduce the conditional number of the feature matrix, we remove features that are constant, giving 741 features. We standardize all features and randomly partition them into 5 subsets for DECO. To compare the performance of all methods, we use 10-fold cross validation and record the prediction error (mean square error, MSE), model size and runtime. The averaged results are summarized in Table \ref{tab:2}. We also report the performance of the null model which predicts the final grade on the test set using the mean final grade in the training set.
\begin{table}[!h]
\centering
\caption{The results of all methods for student performance data with $(n, p, m) = (395, 741, 5)$}
\label{tab:2}
\vskip 0.15in
\begin{tabular}{l|ccc}
\hline
                    & MSE & Model size & runtime\\
\hline\hline
{\bf DECO-3}        & \textbf{3.64} & 1.5          & 37.0\\
{\bf lasso-full}    & 3.79 & 2.2        & 60.8\\
{\bf lasso-refine}  & 3.89 & 2.2        & 70.9\\
{\bf lasso-naive}   & 16.5 & 6.4        & 44.6\\
Null                & 20.7 & ---        & ---\\
\hline
\hline
\end{tabular}
\end{table}

\subsection{Mammalian eye diseases}
This dataset, taken from \citet{scheetz2006regulation}, was collected to study mammalian eye diseases, with gene expression for the eye tissues of 120
twelve-week-old male F2 rats recorded. One gene coded as TRIM32 responsible for causing Bardet-Biedl syndrome is the response of interest. Following the method in \citet{scheetz2006regulation}, 18,976 probes were selected as they exhibited sufficient signal for reliable analysis and at least 2-fold variation in expressions, and we confine our attention to the top 5,000 genes with the
highest sample variance. The 5,000 genes are standardized and partitioned into 100 subsets for DECO. The performance is assessed via 10-fold cross validation following the same approach in Section 5.1. The results are summarized in Table \ref{tab:3}. As a reference, we also report these values for the null model.

\begin{table}[!h]
\centering
\caption{The results of all methods for mammalian eye diseases with $(n, p, m) = (120, 5000, 100)$}
\label{tab:3}
\vskip 0.15in
\begin{tabular}{l|ccc}
\hline
                    & MSE & Model size & runtime\\
\hline\hline
{\bf DECO-3}        & 0.012 & 4.3          & 9.6\\
{\bf lasso-full}    & 0.012 & 11        & 139.0\\
{\bf lasso-refine}  & \textbf{0.010} & 11        & 139.7\\
{\bf lasso-naive}   & 37.65 & 6.8        & 7.9\\
Null                & 0.021 & ---        & ---\\
\hline
\hline
\end{tabular}
\end{table}

\subsection{Electricity load diagram}
This dataset \citep{Trindade:2014} consists of electricity load from 2011 - 2014 for 370 clients. The data are originally recorded in KW for every 15 minutes, resulting in 14,025 attributes. Our goal is to predict the most recent electricity load by using all previous data points. The variance of the 14,025 features ranges from 0 to $10^7$. To reduce the conditional number of the feature matrix, we remove features whose variances are below the lower $10\%$ quantile (a value of $10^5$) and retain 126,231 features. We then expand the feature sets by including the interactions between the first 1,500 attributes that has the largest correlation with the clients' most recent load. The resulting 1,251,980 features are then partitioned into 1,000 subsets for DECO. Because cross-validation is computationally demanding for such a large dataset, we put the first 200 clients in the training set and the remaining 170 clients in the testing set. We also scale the value of electricity load between 0 and 300, so that patterns are more visible. The results are summarized in Table \ref{tab:4}.

\begin{table}[!h]
\centering
\caption{The results of all methods for electricity load diagram data with $(n, p, m) = (370, 1251980, 1000)$}
\label{tab:4}
\vskip 0.15in
\begin{tabular}{l|ccc}
\hline
                    & MSE & Model size & runtime\\
\hline\hline
{\bf DECO-3}        & \textbf{0.691} & 4          & 67.9\\
{\bf lasso-full}    & 2.205 & 6          & 23,515.5\\
{\bf lasso-refine}  & 1.790 & 6          & 22,260.9\\
{\bf lasso-naive}   & $3.6\times 10^8$ & 4966        & 52.9\\
Null                & 520.6 & ---        & ---\\
\hline
\hline
\end{tabular}
\end{table}

\section{Concluding remarks}
In this paper, we have proposed an embarrassingly parallel framework named DECO for distributed estimation. DECO is shown to be theoretically attractive, empirically competitive and is straightforward to implement. In particular, we have shown that DECO achieves the same minimax convergence rate as if the full data were used and the rate does not depend on the number of partitions. We demonstrated the empirical performance of DECO via extensive experiments and compare it to various approaches for fitting full data. As illustrated in the experiments, DECO can not only reduce the computational cost substantially, but often outperform the full data approaches in terms of model selection and parameter estimation.
\par
Although DECO is designed to solve large-$p$-small-$n$ problems, it can be extended to deal with large-$p$-large-$n$ problems by adding a sample space partitioning step, for example, using the \emph{message} approach \citep{wang2014median}. More precisely, we first partition the large-$p$-large-$n$ dataset in the sample space to obtain $l$ row blocks such that each becomes a large-$p$-small-$n$ dataset. We then partition the feature space of each row block  into $m$ subsets. This procedure is equivalent to partitioning the original data matrix $X$ into $l\times m$ small blocks, each with a feasible size that can be stored and fitted in a computer. We then apply the DECO framework to the subsets in the same row block using Algorithm  \ref{alg:1}. The last step is to apply the \emph{message} method to aggregate the $l$ row block estimators to output the final estimate.  This extremely scalable approach will be explored in future work.

\bibliography{reference}

\newpage
\section*{Appendix A: Proof of Theorem \ref{thm:1}}
\subsection*{A.1 Proof of the $\ell_2$ and $\ell_\infty$ convergence}
We first need the following lemmas
\begin{lemma}\label{lemma:3}
Assuming the Condition 3 in Theorem \ref{thm:1}, and defining $\Delta = \hat\beta - \beta_*$, where $\hat\beta$ is the solution to the lasso problem and $\beta_*$ is the true value, then for any set $J\subseteq Q$ ($J$ could be $\emptyset$), where $Q = \{1,2,\cdots, p\}$, we have
\begin{align}
\|\Delta_{J^c}\|_1 \leq 3\|\Delta_J\|_1 + 4\|\beta_{*J^c}\|_1, \label{eq:lemma1.2}
\end{align}
where $\Delta_J$ denotes a sub-vector of $\Delta$ containing coordinates whose indexes belong to $J$ and $\|\Delta_\emptyset\|_1 = 0$.
\end{lemma}
\begin{proof}[Proof of Lemma \ref{lemma:3}]
We follow the proof in \citep{bickel2009simultaneous} and \citep{lounici2008sup}. Define $\hat S(\hat\beta) = \{k:~\hat\beta_k \neq 0\}$. The sufficient and necessary condition (also known as the KKT conditions) for $\hat\beta$ to minimize the \emph{lasso} problem is that
\begin{align*}
\frac{1}{n}(x_i^TY - x_i^TX\hat\beta) &= \lambda_n sign(\hat\beta_i),~ \mbox{for } i\in \hat S(\hat\beta)\\
\frac{1}{n}|x_i^TY - x_i^TX\hat\beta| &\leq \lambda_n, ~ \mbox{for } i\not\in \hat S(\hat\beta).
\end{align*}
Therefore, regardless of $\hat S(\hat\beta)$, the minimizer $\hat\beta$ always satisfies that
\begin{align*}
\frac{1}{n}\|X^TY - X^TX\hat\beta\|_\infty \leq \lambda_n.
\end{align*}
Noticing that $Y = X\beta_* + W$ and $\frac{1}{n}\|X^TW\|_\infty \leq \lambda_n/2$, we have
\begin{align}
\frac{1}{n}\big\|X^TX(\beta_* - \hat\beta)\big\|_\infty \leq \frac{3}{2}\lambda_n. \label{eq:lemma1.1}
\end{align}
At the same time, using the optimality of \emph{lasso} we have
\begin{align*}
\frac{1}{n}\|Y - X\hat\beta\|_2^2 + 2\lambda_n\|\hat\beta\|_1 \leq \frac{1}{n}\|Y - X\beta_*\|_2^2 + 2\lambda_n\|\beta_*\|_1 = \frac{1}{n} \|W\|_2^2 + 2\lambda_n\|\beta_*\|_1,
\end{align*}
which implies
\begin{align*}
2\lambda_n\|\hat\beta\|_1 &\leq 2\lambda_n\|\beta_*\|_1 + \frac{1}{n} \|W\|_2^2 - \frac{1}{n}\|Y - X\hat\beta\|_2^2 \\
 &= 2\lambda_n\|\beta_*\|_1 + \frac{1}{n} \|W\|_2^2 - \frac{1}{n}\|X\beta_* - X\hat\beta + W\|_2^2\\
 &\leq 2\lambda_n\|\beta_*\|_1 + \big|2(\hat\beta - \beta_*)\frac{X^TW}{n}\big|.
\end{align*}
Using $\|\frac{1}{n}X^TW\|_\infty \leq \lambda_n/2$, we know that
\begin{align*}
2\lambda_n\|\hat\beta\|_1 \leq 2\lambda_n\|\beta_*\|_1 + \lambda_n\|\hat\beta - \beta_*\|_1,
\end{align*}
i.e., we have
\begin{align}
2\|\hat\beta\|_1 \leq 2\|\beta_*\|_1 + \|\hat\beta - \beta_*\|_1 = 2\|\beta_*\|_1 + \|\Delta\|_1, \label{eq:lemma3.1}
\end{align}
Let $J$ be any arbitrary subset of $Q$, we have
\begin{align}
2\|\Delta_{J^c}\|_1 = 2\|\hat\beta_{J^c} - \beta_{*J^c}\|_1\leq 2\|\hat\beta_{J^c}\|_1 + 2\|\beta_{*J^c}\|_1. \label{eq:lemma3.2}
\end{align}
Now if $J = \emptyset$, using \eqref{eq:lemma3.1} and \eqref{eq:lemma3.2} we have
\begin{align*}
2\|\Delta\|_1 = 2\|\Delta_{J^c}\|_1 \leq 2\|\hat\beta_{J^c}\|_1 + 2\|\beta_{*J^c}\|_1 = 2\|\hat\beta\|_1 + 2\|\beta_{*}\|_1 \leq 4\|\beta_*\|_1 + \|\Delta\|_1.
\end{align*}
This gives that
\begin{align*}
\|\Delta_{J^c}\|_1 = \|\Delta\|_1 \leq 4\|\beta_*\|_1 = 3\|\Delta_J\|_1 + 4\|\beta_{*J^c}\|_1.
\end{align*}
For $J\neq \emptyset$, because $\ell_1$ norm is decomposable, i.e., $\|\hat\beta\|_1 = \|\hat\beta_J\|_1 + \|\hat\beta_{J^c}\|_1$, using \eqref{eq:lemma3.1}, we have
\begin{align*}
2\|\hat\beta_{J^c}\|_1 + 2\|\beta_{*J^c}\|_1 &= 2\|\hat\beta\|_1 - 2\|\hat\beta_{J}\|_1 + 2\|\beta_{*J^c}\|_1\\
&\leq 2\|\beta_*\|_1 + \|\Delta\|_1 - 2\|\hat\beta_J\|_1 + 2\|\beta_{*J^c}\|_1\\
&= 2\|\beta_{*J}\|_1 + 2\|\beta_{*J^c}\|_1 + \|\Delta_J\|_1 + \|\Delta_{J^c}\|_1 - 2\|\hat\beta_J\|_1 + 2\|\beta_{*J^c}\|_1\\
&= 2(\|\beta_{*J}\|_1 - \|\hat\beta_J\|_1) + \|\Delta_J\|_1 + \|\Delta_{J^c}\|_1 + 4\|\beta_{*J^c}\|_1\\
&\leq 3\|\Delta_J\|_1 + \|\Delta_{J^c}\|_1 + 4\|\beta_{*J^c}\|_1,
\end{align*}
where the second inequality is due to \eqref{eq:lemma3.1}. Thus, combining the above result with \eqref{eq:lemma3.2} we have proved that
\begin{align*}
\|\Delta_{J^c}\|_1 \leq 3\|\Delta_J\|_1 + 4\|\beta_{*J^c}\|_1.
\end{align*}
\end{proof}

\begin{lemma}\label{lemma:4}
Assume the Condition 1 and 2 in Theorem \ref{thm:1}. For any $J\subseteq\{1, 2, \cdots, p\}$ ($J$ could be $\emptyset$) and $|J| \leq s$ and any $v\in \mathcal{R}^p$ such that $\|v_{J^c}\|_1 \leq c_0 \|v_{J}\|_1 + c_1\|\beta_{*J^c}\|_1$, we have
\begin{align}
\frac{1}{n}\|Xv\|_2^2 \geq \bigg(M_1 - \frac{1 + 2c_0}{\gamma_1}\bigg)\|v_J\|_2^2 - 2c_1\gamma_2\sqrt{s}\lambda_n^{q}\|\beta_{*J^c}\|_1 \|v_J\|_2, \label{eq:lemma2.1}
\end{align}
where $v_J$ denotes a sub-vector of $v$ containing coordinates whose indexes belong to $J$.
\end{lemma}
\begin{proof}[Proof of Lemma \ref{lemma:4}]
When $J = \emptyset$, the result is straightforward and thus omitted. Assume $|J| > 0$. For convenience, we define $\tilde v$ to be the vector that extends $v_J$ to $p$-dimensional by adding zero coodinates, i.e., 
\begin{align*}
\tilde v_i &= v_i\quad\mbox{if } i \in J\\
\tilde v_i &= 0\quad\mbox{if } i \not\in J
\end{align*}
We use $v_J^{(i)}$ to denote the $i^{th}$ coordinate of $v_J$.
For any $J\subseteq\{1, 2, \cdots, p\}$ with $|J| = s$ and any $v\in \mathcal{R}^p$ such that $\|v_{J^c}\|_1 \leq c_0 \|v_{J}\|_1 + c_1\|\beta_{*J^c}\|_1$, we have
\begin{align*}
\frac{\|X\tilde v\|_2^2}{n\|v_J\|_2^2} &= \frac{\tilde v^T(X^TX/n - M_1 I_p)\tilde v}{\|v_J\|_2^2} + M_1\\
&= M_1 + \frac{\sum_{i = 1}^p (x_i^Tx_i/n - M_1) |\tilde v_i|^2 + \sum_{i\neq j} (x_i^Tx_j/n) \tilde v_i\tilde v_j}{\|v_J\|_2^2}\\
&= M_1 + \frac{\sum_{i \in J} (x_i^Tx_i/n - M_1) |v_J^{(i)}|^2 + \sum_{i\neq j\in J} (x_i^Tx_j/n) v_J^{(i)} v_J^{(i)}}{\|v_J\|_2^2}\\
&\geq M_1 - \frac{1}{\gamma_1 s}\sum_{i\neq j}\frac{v_{J}^{(i)}v_{J}^{(j)}}{\|v_J\|_2^2} \geq M_1 - \frac{1}{\gamma_1 s}\frac{\|v_J\|_1^2}{\|v_J\|_2^2}.
\end{align*}
Notice that $\|v_J\|_1^2\leq s\|v_J\|_2^2$ because $|J|\leq s$. Thus, we have
\begin{align*}
\frac{1}{n}\|Xv\|_2^2 &\geq \frac{1}{n}\|X\tilde v\|_2^2 + 2\tilde v^T(\frac{1}{n}X^TX)(v - \tilde v)\\
&\geq M_1\|v_J\|_2^2 - \frac{1}{\gamma_1 s}\|v_J\|_1^2 - 2\max_{i\neq j}\frac{1}{n}|x_i^Tx_j|\|v_J\|_1\|v_{J^c}\|_1\\
&\geq \bigg(M_1 - \frac{1}{\gamma_1}\bigg)\|v_J\|_2^2 - 2c_0\max_{i\neq j}\frac{1}{n}|x_i^Tx_j|\|v_J\|_1^2 - 2c_1\max_{i\neq j}\frac{1}{n}|x_i^Tx_j|\|\beta_{*J^c}\|_1\|v_J\|_1\\
&\geq \bigg(M_1 - \frac{1}{\gamma_1}\bigg)\|v_J\|_2^2 - \frac{2c_0}{\gamma_1 s}\|v_J\|_1^2 - 2c_1\gamma_2\lambda_n^q\beta_{*J^c}\|_1\|v_J\|_1\\
&\geq \bigg(M_1 - \frac{1 + 2c_0}{\gamma_1}\bigg)\|v_J\|_2^2 - 2c_1\gamma_2\sqrt{s}\lambda_n^q\|\beta_{*J^c}\|_1\|v_J\|_2.
\end{align*}
\end{proof}

\begin{lemma}\label{lemma:5}
Assume the Condition 1 and 2 in Theorem \ref{thm:1}. For any $J\subseteq\{1, 2, \cdots, p\}$ ($J$ could be $\emptyset$) and $|J| \leq s$ and any $v\in \mathcal{R}^p$ such that $\|v_{J^c}\|_1 \leq c_0 \|v_{J}\|_1 + c_1\|\beta_{*J^c}\|_1$, we have
\begin{align}
\frac{1}{n}\|Xv\|_2^2 \geq \bigg(M_1 - \frac{2(1 + c_0)^2}{\gamma_1}\bigg) \|v\|_2^2 - 2c_1^2 \lambda_n^q \|\beta_{*J^c}\|_1^2. \label{eq:lemma5}
\end{align}
where $v_J$ denotes a sub-vector of $v$ containing coordinates whose indexes belong to $J$.
\end{lemma}
\begin{proof}[Proof of Lemma \ref{lemma:5}]
Different from Lemma \ref{lemma:4}, we have
\begin{align*}
\frac{1}{n}\|Xv\|_2^2 &\geq \sum_{i\in Q}\frac{1}{n} \|x_i\|_2^2 v_i^2 + \sum_{i\neq j\in Q} \frac{1}{n} x_i^Tx_j v_iv_j\\
&\geq M_1\|v\|_2^2 - \max_{i\neq j}\frac{1}{n}|x_i^Tx_j| \|v\|_1^2 = M_1\|v\|_2^2 - \max_{i\neq j}\frac{1}{n}|x_i^Tx_j|(\|v_J\|_1 + \|v_{J^c}\|_1)^2\\
&\geq M_1\|v\|_2^2 - \max_{i\neq j}\frac{1}{n}|x_i^Tx_j|\bigg((1 + c_0)\|v_J\|_1 + c_1\|\beta_{*J^c}\|_1\bigg)^2\\
&\geq M_1\|v\|_2^2 - 2\max_{i\neq j}\frac{1}{n}|x_i^Tx_j| (1 + c_0)^2\|v_J\|_1^2 - 2\max_{i\neq j}\frac{1}{n}|x_i^Tx_j| c_1^2\|\beta_{*J^c}\|_1^2\\
&\geq M_1\|v\|_2^2 - \frac{2(1 + c_0)^2}{\gamma_1} \|v_J\|_2^2 - 2c_1^2\gamma_2 \lambda_n^q \|\beta_{*J^c}\|_1^2\\
&\geq \bigg(M_1 - \frac{2(1 + c_0)^2}{\gamma_1}\bigg) \|v\|_2^2 - 2c_1^2\gamma_2 \lambda_n^q \|\beta_{*J^c}\|_1^2
\end{align*}
\end{proof}

Now, We turn to the proof of $\ell_2$ and $\ell_\infty$ convergence in Theorem \ref{thm:1}. 
\begin{proof}[{\bf (Partial) proof of Theorem \ref{thm:1}}]
According to Lemma \ref{lemma:3}, \ref{lemma:4}, \ref{lemma:5} and \eqref{eq:lemma1.2}, \eqref{eq:lemma1.1} and \eqref{eq:lemma2.1}, we have
\begin{align}
\big\|\frac{1}{n}X^TX\Delta\big\|_\infty \leq \frac{3}{2}\lambda_n\label{eq:thm1.1}
\end{align}
and
\begin{align}
\|\Delta\|_1\leq 4\|\Delta_J\|_1 + 4\|\beta_{*J^c}\|_1 \leq 4\sqrt{s}\|\Delta_{J}\|_2 + 4\|\beta_{*J^c}\|_1\label{eq:thm1.2}
\end{align}
and
\begin{align}
\frac{1}{n}\|X\Delta\|_2^2 \geq \bigg(M_1 - \frac{7}{\gamma_1}\bigg)\|\Delta_{J}\|_2^2 - 8\gamma_2\sqrt{s}\lambda_n^{q}\|\beta_{*J^c}\|_1\|\Delta_J\|_2.\label{eq:thm1.3}
\end{align}
Using Equations \eqref{eq:thm1.1} and \eqref{eq:thm1.2}, we have
\begin{align*}
\frac{1}{n}\|X\Delta\|_2^2 \leq \|\frac{1}{n}X^TX\Delta\|_\infty\|\Delta\|_1 \leq 6\lambda_n\sqrt{s}\|\Delta_{J}\|_2 + 6\lambda_n\|\beta_{*J^c}\|_1,
\end{align*}
which combining with \eqref{eq:thm1.3} implies that
\begin{align*}
\bigg(M_1 - \frac{7}{\gamma_1}\bigg)\|\Delta_{J}\|_2^2 - 2(3\sqrt{s}\lambda_n + 4\gamma_2\sqrt{s}\lambda_n^{q}\|\beta_{*J^c}\|_1)\|\Delta_J\|_2 - 6\lambda_n\|\beta_{*J^c}\|_1\leq 0
\end{align*}
This is a quadratic form and with some simple algebra, we get a loose solution to the quadratic inequality
\begin{align*}
\frac{1}{2}\bigg(M_1 - \frac{7}{\gamma_1}\bigg)\|\Delta_{J}\|_2^2 \leq \frac{2(3\sqrt{s}\lambda_n + 4\gamma_2\sqrt{s}\lambda_n^{q}\|\beta_{*J^c}\|_1)^2}{M_1 - \frac{7}{\gamma_1}} + 6\lambda_n\|\beta_{*J^c}\|_1,
\end{align*}
thus
\begin{align*}
\|\Delta_{J}\|_2^2\leq \frac{72\gamma_1^2 s}{(M_1\gamma_1 - 7)^2}\lambda_n^2 + \frac{192\gamma_1^2\gamma_2^2\|\beta_{*J^c}\|_1^2s}{(M_1\gamma_1 - 7)^2}\lambda_n^{2q} + \frac{12\gamma_1\|\beta_{*J^c}\|_1}{M_1\gamma_1 - 7}\lambda_n,
\end{align*}
and thus
\begin{align}
\notag \|\Delta_{J}\|_2 &\leq \sqrt{\frac{72\gamma_1^2 s}{(M_1\gamma_1 - 7)^2}\lambda_n^2 + \frac{192\gamma_1^2\gamma_2^2\|\beta_{*J^c}\|_1^2s}{(M_1\gamma_1 - 7)^2}\lambda_n^{2q} + \frac{12\gamma_1\|\beta_{*J^c}\|_1}{M_1\gamma_1 - 7}\lambda_n} \\
&\leq \frac{6\sqrt{2}\gamma_1}{M_1\gamma_1 - 7}\sqrt{s}\lambda_n + \frac{8\sqrt{3}\gamma_1\gamma_2\|\beta_{*J^c}\|_1}{M_1\gamma_1 - 7}\sqrt{s}\lambda_n^{q} + \frac{2\sqrt{3}\gamma_1^{\frac{1}{2}}\|\beta_{*J^c}\|_1^{\frac{1}{2}}}{\sqrt{M_1\gamma_1 - 7}}\lambda_n^{\frac{1}{2}} \label{eq:thm1.4}
\end{align}
Similarly, for $\|\Delta\|_2^2$, using \eqref{eq:lemma5} we have
\begin{align*}
\bigg(M_1 - \frac{32}{\gamma_1}\bigg) \|\Delta\|_2^2 - 32\gamma_2 \lambda_n^q \|\beta_{*J^c}\|_1^2\leq \frac{1}{n}\|X\Delta\|_2^2\leq 6\lambda_n\sqrt{s}\|\Delta_J\|_2 + 6\lambda_n\|\beta_{*J^c}\|_1.
\end{align*}
Noticing that $\|\Delta_J\|_2\leq \|\Delta\|_2$, we can solve the quadratic inequality and obtain that
\begin{align}
\|\Delta\|_2^2\leq \frac{18\gamma_1^2 s\lambda_n^2}{(M_1\gamma_1 - 32)^2} + 6\lambda_n\|\beta_{*J^c}\|_1 + 32\gamma_2\lambda_n^q\|\beta_{*J^c}\|_1^2.
\end{align}
For the sup-norm, we make use of \eqref{eq:thm1.4}. Notice that
\begin{align*}
e_j^T\frac{X^TX}{n}\Delta = \frac{x_j^TX}{n}\Delta = \frac{\|x_j\|_2^2}{n}\Delta_j + \sum_{i\neq j} \frac{x_i^Tx_j}{n}\Delta_i
\end{align*}
which combning with \eqref{eq:thm1.1} and \eqref{eq:thm1.2} implies that
\begin{align*}
\frac{\|x_j\|_2^2}{n}|\Delta_j| &\leq \bigg|e_j^T\frac{X^TX}{n}\Delta\bigg| + \bigg|\sum_{i\neq j} \frac{x_i^Tx_j}{n}\Delta_i\bigg| \leq \|\frac{1}{n}X^TX\Delta\|_\infty + \max_{i\neq k} \frac{1}{n}|x_i^Tx_k| \|\Delta\|_1\\
&\leq \frac{3}{2}\lambda_n + 4\max_{i\neq k} \frac{1}{n}|x_i^Tx_k|\sqrt{s}\|\Delta_J\|_2 + 4\max_{i\neq k} \frac{1}{n}|x_i^Tx_k|\|\beta_{*J^c}\|_1
\end{align*}
Note that $\max_{i\neq k}\frac{1}{n}|x_i^Tx_k|\leq \min\{\frac{1}{\gamma_1 s}, ~\gamma_2\lambda_n^{q}\}$ also implies that $\max_{i\neq k}\frac{1}{n}|x_i^Tx_k|\leq \sqrt{\frac{\gamma_2\lambda_n^{q}}{\gamma_1 s}}$.
Therefore, using result in \eqref{eq:thm1.4} we have
\begin{align*}
M_1\|\Delta\|_\infty\leq \frac{3}{2}\lambda_n + \frac{24\sqrt{2}}{M_1\gamma_1 - 7}\lambda_n + \frac{32\sqrt{3}\gamma_2}{M_1\gamma_1 - 7}\|\beta_{*J^c}\|_1\lambda_n^{q} + \frac{8\sqrt{3}\gamma_2^{\frac{1}{2}}}{\sqrt{M_1\gamma_1 - 7}}\|\beta_{*J^c}\|_1^{\frac{1}{2}}\lambda_n^{\frac{1 + q}{2}} + 4\gamma_2\|\beta_{*J^c}\|_1\lambda_n^{q},
\end{align*}
which yields,
\begin{align*}
\|\Delta\|_\infty \leq \frac{3M_1\gamma_1 + 51}{2M_1(M_1\gamma_1 - 7)}\lambda_n + \frac{4M_1\gamma_1\gamma_2 + 36\gamma_2}{M_1(M_1\gamma_1 - 7)}\|\beta_{*J^c}\|_1\lambda_n^{q} + \frac{8\sqrt{3\gamma_2}}{M_1\sqrt{M_1\gamma_1 - 7}}\|\beta_{*J^c}\|_1^{\frac{1}{2}}\lambda_n^{\frac{1 + q}{2}}.
\end{align*}
This completes the proof.
\end{proof}

\subsection*{A.2 Proof of the sign consistency}
Our conclusion on sign consistency is stated as follows
\begin{theorem}\label{thm:4}
Let $J$ be the set containing indexes of all the nonzero coefficients. Assume all the conditions in Theorem \ref{thm:1}. In addition, if the following conditions hold
\begin{align*}
\min_{k\in J}|\beta_k| \geq \frac{2}{M_1}\lambda_n,
\end{align*}
then the solution to the \emph{lasso} is unique and satisfies the sign consistency, i.e, 
\begin{align*}
sign(\hat\beta_k) = sign(\beta_{*k}),~\forall k\in J\quad\mbox{and}\quad \hat\beta_k = 0, ~\forall k\in J^c.
\end{align*}
\end{theorem}
Here we use the primal-dual witness (PDW) approach \citep{wainwright2009sharp} to prove sign consistency. The PDW approach works on the following two terms
\begin{align*}
Z_k = \frac{1}{n\lambda_n}x_{k}^T\Pi_{X_J^\perp}W + \frac{1}{n}x_{k}^TX_J\bigg(\frac{1}{n}X_J^TX_J\bigg)^{-1}\breve{z}_J,
\end{align*}
where $\Pi_A$ is the projection on to the linear space spanned by the vectors in $A$ and
\begin{align*}
\Delta_k = e_k^T\bigg(\frac{1}{n}X_J^TX_J\bigg)^{-1}\bigg(\frac{1}{n}X_J^TW - \lambda_n sign(\beta_{*J})\bigg),
\end{align*}
for which \citep{wainwright2009sharp} proves the following lemma
\begin{lemma}\citep{wainwright2009sharp} \label{lemma:5.5}
If $Z_k$ and $\Delta_k$ satisfy that
\begin{align*}
sign(\beta_{*k} + \Delta_k) = sign(\beta_{*k}), ~\forall k\in J\quad\mbox{and}\quad |Z_k|< 1, ~\forall k\in J^c,
\end{align*}
then the optimal solution to \emph{lasso} is unique and satisfies the sign consistency, i.e.,
\begin{align*}
sign(\hat\beta_k) = sign(\beta_{*k}),~\forall k\in J\quad\mbox{and}\quad \hat\beta_k = 0, ~\forall k\in J^c.
\end{align*}
\end{lemma}
Therefore, we just need to verify the two conditions in Lemma \ref{lemma:5.5} for Theorem \ref{thm:4}. Before we proceed to prove Theorem \ref{thm:4}, we state another lemma that is needed for the proof.
\begin{lemma}\citep{varah1975lower}\label{lemma:6}
Let $A$ be a strictly diagonally dominant matrix and define $\delta = \min_k(|A_{kk}| - \sum_{j\neq k}|A_{kj}|) > 0$, then we have
\begin{align*}
\|A^{-1}\|_\infty \leq \delta^{-1},
\end{align*}
where $\|A\|_\infty$ is the maximum of the row sums of $A$.
\end{lemma}
\begin{proof}[{\bf Proof of Theorem \ref{thm:4}}]
We first bound $|Z_k|$ for $k \in J^c$. Notice the first term in $Z_k$ follows that
\begin{align*}
\frac{1}{n\lambda_n}x_k^T\Pi_{X_J^\perp}W = \frac{1}{n\lambda_n}x_k^TW - \frac{1}{n\lambda_n}x_k^TX_J(X_J^TX_J)^{-1}X_J^TW,
\end{align*}
where $\frac{1}{n\lambda_n}x_k^TW$ follows
\begin{align*}
\bigg|\frac{1}{n\lambda_n}x_k^TW\bigg| \leq \frac{1}{\lambda_n} \|\frac{1}{n}X^TW\|_\infty \leq \frac{1}{2}
\end{align*}
and $\frac{1}{n\lambda_n}x_k^TX_J(X_J^TX_J)^{-1}X_J^TW$ follows
\begin{align*}
\bigg|\frac{1}{n\lambda_n}x_k^TX_J(X_J^TX_J)^{-1}X_J^TW\bigg|&\leq \frac{1}{\lambda_n}\|\frac{1}{n}x_k^TX_J\|_1\|(X_J^TX_J)^{-1}X_J^TW\|_\infty\\
\end{align*}
From Condition 2 in Theorem \ref{thm:1}, we know that 
\begin{align*}
\|\frac{1}{n}x_k^TX_J\|_1 \leq \sum_{j\in J} \frac{1}{n}|x_k^Tx_j| \leq \frac{1}{\gamma_1}
\end{align*}
and using Lemma \ref{lemma:6} we have
\begin{align*}
\|(X^TX/n)^{-1}\|_\infty = \max_{k\in Q}\|e_k^T(X_J^TX_J/n)^{-1}\|_1\leq (M_1 - 1/\gamma_1)^{-1}.
\end{align*}
Thus, we have
\begin{align*}
\frac{1}{\lambda_n}\|(X_J^TX_J)^{-1}X_J^TW\|_\infty \leq \frac{1}{\lambda_n}\|(X_J^TX_J/n)^{-1}\|_\infty\|\frac{1}{n}X_J^TW\|_\infty \leq\frac{\gamma_1}{2(M_1\gamma_1 - 1)}.
\end{align*}
Together, the first term can be bounded as
\begin{align}
\bigg|\frac{1}{n\lambda_n}x_k^T\Pi_{X_J^\perp}W\bigg|\leq \frac{1}{2} + \frac{1}{2(M_1\gamma_1 - 1)}. \label{eq:term2}
\end{align}
The second term can be bounded similarly as the first term, i.e., 
\begin{align*}
\bigg|\frac{1}{n}x_k^TX_J(X_J^TX_J)^{-1}\breve{z}_J\bigg|\leq \|\frac{1}{n}x_k^TX_J\|_1\|(X_J^TX_J)^{-1}\breve{z}_J\|_\infty\leq\frac{1}{M_1\gamma_1 - 1},
\end{align*}
Therefore, we have
\begin{align*}
|Z_k|\leq \frac{1}{2} + \frac{3}{2(M_1\gamma_1 - 1)}.
\end{align*}
It is easy to see that when $\gamma_1 > 32/M_1$, we have
\begin{align*}
|Z_k| < 1,~\forall k\in J^c
\end{align*}
and completes the proof for $Z_k$. We now turn our attention to $\Delta_k$ and check whether $sign(\beta_{*k}) = sign(\beta_{*k} + \Delta_k)$. For $\Delta_k$, we have
\begin{align*}
|\Delta_k| &= \bigg|e_k^T\bigg(\frac{1}{n}X_J^TX_J\bigg)^{-1} \bigg(\frac{1}{n}X_J^TW - \lambda_n sign(\beta_{*J})\bigg)\bigg|\\
&\leq \bigg|e_k^T\bigg(\frac{1}{n}X_J^TX_J\bigg)^{-1}\frac{X_J^TW}{n}\bigg| + \lambda_n\bigg\|\bigg(\frac{1}{n}X_J^TX_J\bigg)^{-1}\bigg\|_\infty\\
&\leq \bigg\|\bigg(\frac{1}{n}X_J^TX_J\bigg)^{-1}\bigg\|_\infty\|X_J^TW/n\|_\infty + \lambda_n\bigg\|\bigg(\frac{1}{n}X_J^TX_J\bigg)^{-1}\bigg\|_\infty\\
&\leq\frac{\gamma_1}{2(M_1\gamma_1 - 1)}\lambda_n + \frac{\gamma_1}{M_1\gamma_1 - 1}\lambda_n\\
&= \frac{3\gamma_1}{2(M_1\gamma_1 - 1)}\lambda_n.
\end{align*}
Thus, with the conditions in Theorem 2, we have
\begin{align*}
|\Delta_k|\leq \frac{3\gamma_1}{2(M_1\gamma_1 - 1)}\lambda_n = \frac{3}{2(M_1 - 1/\gamma_1)}\lambda_n < \frac{2}{M_1}\lambda_n.
\end{align*}
To meet the requirement $sign(\beta_{*k}) = sign(\beta_{*k} + \Delta_k)$, we just need $\min_{k\in J}|\beta_k| \geq \frac{2}{M_1}\lambda_n$ and this completes the proof.
\end{proof}

\section*{Appendix B: Proof of Corollary \ref{co:1} and \ref{co:2}}
To prove the two corollaries, we just need to adapt the magnitude of $\max_{i\neq j}\frac{1}{n}|x_i^Tx_j|$ to the correct order.
\begin{proof}[{\bf Proof of Corollary \ref{co:1} and \ref{co:2}}]
To prove Corollary \ref{co:1}, we just need to take $\gamma_2$ arbitrarily large and $q = 1$. The result follows immediately from Theorem \ref{thm:1}.
\par
To prove Corollary \ref{co:2}, we first determine the set $J$ by taking the larger signals as follows
\begin{align*}
J = \{k:~|\beta_k|\geq \lambda_n\}.
\end{align*}
Then the size of $J$ can be bounded as
\begin{align*}
s = |J| \leq \frac{R}{\lambda_n^r}
\end{align*}
and the size of $\|\beta_*J^c\|_1$ can be bounded as
\begin{align*}
\|\beta_*J^c\|_1 = \sum_{k\in J^c} |\beta_*k| \leq \lambda_n^{1 - r}\sum_{k\in J^c} |\beta_*k|^r\leq R\lambda_n^{1 - r}.
\end{align*}
Now we take $\gamma_2 = 1/\|\beta_*J^c\|_1$ and $q = 1$, then the bound on $\max_{i\neq j}\frac{1}{n}|x_i^Tx_j|$ becomes
\begin{align*}
\max_{i\neq j}\frac{1}{n}|x_i^Tx_j|\leq\min\bigg\{\frac{1}{\gamma_1 s}, \frac{\lambda_n}{\|\beta_*J^c\|_1}\bigg\} \leq \min\bigg\{\frac{\lambda_n^r}{\gamma_1 R}, \frac{\lambda_n^r}{R}\bigg\} \leq \frac{\lambda_n^r}{\gamma_1 R},
\end{align*}
which completes the proof.
\end{proof}

\section*{Appendix C: Proof of Lemma \ref{lemma:1} and \ref{lemma:2}}
We first fully state Lemma \ref{lemma:1} here
\begin{lemma}\citep{wang2015consistency}\label{lemma:1}
Assuming $X\sim N(0, \Sigma)$ and $p > c_0n$ for some $c_0 > 1$, we have that for any $C > 0$, there exists some constant $0 < c_1 < 1 < c_2$ and $c_3 > 0$ such that for any $i \neq j\in Q$
\begin{align*}
\mathbb{P}\bigg(\frac{1}{n}|\tilde x_i|_2^2 < \frac{c_1c_*}{c^*}\bigg)\leq 2e^{-Cn},\quad \mathbb{P}\bigg(\frac{1}{n}|\tilde x_i|_2^2 > \frac{c_2c^*}{c_*}\bigg) \leq 2e^{-Cn},
\end{align*}
and
\begin{align*}
\mathbb{P}\bigg(\frac{1}{n}|\tilde x_i^T\tilde x_j| > \frac{c_4c^*t}{c_*}\frac{1}{\sqrt{n}}\bigg)\leq 5e^{-Cn} + 2e^{-t^2/2},
\end{align*}
for any $t > 0$, where $c_4 = \sqrt{\frac{c_2(c_0 - c_1)}{c_3(c_0 - 1)}}$ and $c_*, c^*$ are the smallest and largest eigenvalues of $\Sigma$.
\end{lemma}
Lemma \ref{lemma:1} and the first part of \ref{lemma:2} are existing results from \citep{wang2015consistency} and \citep{wang2015high}. We focus on proving the second part of Lemma \ref{lemma:2}.
\begin{proof}[{\bf Proof of Lemma \ref{lemma:1} and \ref{lemma:2}}]
Lemma \ref{lemma:1} follows immediately from Lemma 3 in \citep{wang2015consistency} and the first part of Lemma \ref{lemma:2} follows Lemma 4 in \citep{wang2015consistency}.
\par
To prove the second part of Lemma \ref{lemma:2}, we first define $H = X^T(XX^T)^{-\frac{1}{2}}$. When $X\sim N(0, \Sigma)$, $H$ follows the $MACG(\Sigma)$ distribution as indicated in Lemma 3 in \citep{wang2015consistency} and Theorem 1 in \citep{wang2015high}. For simplicity, we only consider the case where $k = 1$. 
\par
For vector $v$ with $v_1 = 0$, we define $v' = (v_2, v_3, \cdots, v_p)^T$ and we can always identify a $(p - 1)\times (p - 1)$ orthogonal matrix $T'$ such that $T'v' = \|v'\|_2e_1'$ where $e_1'$ is a $(p - 1) \times 1$ unit vector with the first coordinate being 1. Now we define a new orthogonal matrix $T$ as
\begin{align*}
T = \begin{pmatrix}
1 & 0\\
0 & T'
\end{pmatrix}
\end{align*}
and we have
\begin{align*}
Tv = \begin{pmatrix}
1 & 0\\
0 & T'
\end{pmatrix}
\begin{pmatrix}
0 \\
v'
\end{pmatrix}
=
\begin{pmatrix}
0\\
\|v\|_2e_1'
\end{pmatrix}
= \|v\|_2 e_2.
\quad\mbox{and}\quad
e_1^T T^T = e_1^T\begin{pmatrix}
1 & 0\\
0 & T^{'T}
\end{pmatrix}
= e_1^T
\end{align*}
Therefore, we have
\begin{align*}
e_1^THH^Tv = e_1^TT^TTHH^TT^TTv = e_1^TT^THH^TT^Te_2 = \|v\|_2 e_1^T\tilde H\tilde H^Te_2.
\end{align*}
Since $H$ follows $MACG(\Sigma)$, $\tilde H = T^TH$ follows $MACG(T^T\Sigma T)$ for any fixed $T$. Therefore, we can apply Lemma 3 in \citep{wang2015consistency} or Lemma \ref{lemma:1} again to obtain that
\begin{align*}
\Pr\bigg(&|e_1^TX^T(XX^T)^{-1}Xv|\geq \frac{\|v\|_2 c_4c^*t}{c_*}\frac{\sqrt{n}}{p}\bigg) = \Pr\bigg(|e_1^THH^Tv|\geq \frac{\|v\|_2 c_4c^*t}{c_*}\frac{\sqrt{n}}{p}\bigg)\\
 & = \Pr\bigg(\|v\|_2|e_1^T\tilde H\tilde H^Te_2|\geq \frac{\|v\|_2 c_4c^*t}{c_*}\frac{\sqrt{n}}{p}\bigg) = \Pr\bigg(|e_1^T\tilde H\tilde H^Te_2| \geq \frac{c_4c^*t}{c_*}\frac{\sqrt{n}}{p}\bigg)\leq 5e^{-Cn} + 2e^{-t^2/2}.
\end{align*}
Applying the above result to $v = (0, \beta_*^{(-1)})$ we have
\begin{align*}
\frac{1}{n}|\tilde x_1^T\tilde X^{(-1)}\beta_*^{(-1)}| = \frac{1}{n}|e_1\tilde X^T\tilde X v| = \frac{1}{n}\bigg|e_1 X^T\bigg(\frac{XX^T}{p}\bigg)^{-1}X v\bigg| = \frac{p}{n}|e_1 X^T(XX^T)^{-1}X v|\leq \frac{c_4c^*t}{c_*}\frac{\|\beta_*\|_2}{\sqrt{n}},
\end{align*}
with probability at least $1 - 5e^{-Cn} - 2e^{-t^2/2}$.
\par
In addition, we know that $\sigma_0^2 = var(Y) = \beta_*^T\Sigma\beta_* + \sigma^2$ and thus
\begin{align*}
\|\beta_*\|_2 \leq \sqrt{\frac{\sigma_0^2 - \sigma^2}{c_*}}.
\end{align*}
Consequently, we have
\begin{align*}
\Pr\bigg(\frac{1}{n}|\tilde x_1^T\tilde X^{(-1)}\beta_*^{(-1)}| \geq \frac{\sqrt{\sigma_0^2 - \sigma^2} t}{\sqrt{n}}\bigg) \leq 2\exp\bigg(-\frac{c_*^3}{2c_4^2c^{*2}}t^2\bigg) + 5e^{-Cn}.
\end{align*}
Applying the result to any $k\in Q$ and taking the union bound gives the result in Lemma \ref{lemma:2}.
\end{proof}

\section*{Appendix D: Proof of Theorem \ref{thm:2} and \ref{thm:3}}
We assemble all previous results to prove these two theorems.
\begin{proof}[{\bf Proof of Theorem \ref{thm:2} and \ref{thm:3}}]
We just need to verify the Condition 1 and 3 listed in Theorem \ref{thm:1} and the variants of Condition 2 in two corollaries.
\par
First, we verify Condition 1. Taking $M_1 = \frac{c_1c_*}{c^*}$ and $M_2 =
\frac{c_2c^*}{c_*}$ and using Lemma \ref{lemma:1}, we have that
\begin{align*}
\Pr\bigg(M_1\leq \frac{1}{n}|\tilde x_i^T\tilde x_i| \leq M_2,~\forall
i\in Q\bigg)\geq 1 - 4p e^{-Cn}.
\end{align*}

\par
Next, we verify Condition 3, which follows immediately from Lemma \ref{lemma:2}. For any $l \in\{1,2,3,\cdots, m\}$, we have
\begin{align*}
\max_l\frac{1}{n}\|\tilde X^{(l)}\tilde W^{(l)}\|_\infty \leq \max_{k\in Q}\frac{1}{n}\big|\tilde x_k^T \tilde X^{(-k)}\beta_{*}^{(-k)}\big| + \max_{k\in Q} \frac{1}{n}|\tilde x_k^T\tilde \varepsilon| \leq \frac{\sqrt{2}\sigma_0 t}{\sqrt{n}},
\end{align*}
with probability at least $1 - 2p\exp\bigg(-\frac{c_*c_0^2}{2c^*c_2(1 - c_0)^2} t^2\bigg) - 2p\exp\bigg(-\frac{c_*^3}{2c_4^2c^{*2}}t^2\bigg) - 9pe^{-Cn}$. Taking $t = A\sqrt{\log p}/(2\sqrt{2})$ for any $A>0$, we have
\begin{align*}
\Pr\bigg(\max_{l}\frac{1}{n}\|\tilde X^{(l)}\tilde W^{(l)}\|_\infty\geq \frac{1}{2}A\sigma_0\sqrt{\frac{\log p}{n}}\bigg) \leq 2p^{1 - C_1A^2} + 4p^{1 - C_2A^2} + 9pe^{-Cn},
\end{align*}
where $C_1 = \frac{c_*c_0^2}{16c^*c_2(1 - c_0)^2}$ and $C_2 = \frac{c_*^3}{16c_4^2c^{*2}}$. This also indicates that $\lambda_n$ should be chosen as
\begin{align*}
\lambda_n = A\sigma_0\sqrt{\frac{\log p}{n}}.
\end{align*}
\par
Finally, we verify the two conditions in Corollary \ref{co:1} and \ref{co:2}. Notice that Lemma \ref{lemma:1} indicates that
\begin{align*}
\Pr\bigg(\max_{i\neq j}\frac{1}{n}|\tilde x_i^T\tilde x_j| \geq A\sqrt{\frac{\log p}{n}}\bigg) \leq 2p^{1 - 8C_2A^2/c_*} + 5p e^{-Cn} \leq 2p^{1 - C_2A^2} + 5p e^{-Cn} .
\end{align*}
Therefore, the two conditions in Corollary \ref{co:1} and \ref{co:2} will be satisfied as long as
\begin{align*}
A^2\gamma_1^2s^2\frac{\log p}{n} \leq 1\quad\mbox{and}\quad A^2\gamma_1^2R^2\bigg(\frac{\log p}{n}\bigg)^{1 - r}\leq 1.
\end{align*}
Now we have verified that the three conditions hold for all subsets of the data. Let $\hat\beta^{(l)}$ and $\beta_*^{(l)}$ denote the estimate and true value of the coefficients on the $l^{th}$ worker and define $s_l = \|\beta_*^{(l)}\|_0$ and $R_l = \|\beta_*^{(l)}\|_r^r$. Applying Corollary \ref{co:1} and \ref{co:2} to each subset and taking $\gamma_1 = 64/M_1$ we have
\begin{align*}
\|\hat\beta^{(l)} - \beta_*^{(l)}\|_\infty\leq \frac{5A\sigma_0}{M_1}\sqrt{\frac{\log p}{n}}\quad\mbox{and}\quad \|\hat\beta^{(l)} - \beta_*^{(l)}\|_2^2\leq \frac{72A^2\sigma_0^2}{M_1^2} \frac{s_l\log p}{n}
\end{align*}
for $l = 1, 2,\cdots, m$ and $\beta_*$ being s-sparse. For $\beta_*\in\mathbb{B}(r, R)$, we have
\begin{align*}
\|\hat\beta^{(l)} - \beta_*^{(l)}\|_\infty\leq \frac{12 A\sigma_0}{M_1}\sqrt{\frac{\log p}{n}}\quad\mbox{and}\quad \|\hat\beta^{(l)} - \beta_*^{(l)}\|_2^2\leq \bigg(\frac{72}{M_1^2} + 38\bigg) (A\sigma_0)^{2-r} R_l\bigg(\frac{\log p}{n}\bigg)^{1 - \frac{r}{2}}.
\end{align*}
Notice that $\|\hat\beta - \beta_*\|_2^2 = \sum_{l = 1}^m \|\hat\beta^{(l)} - \beta_*^{(l)}\|_2^2$, $s = \sum_{l = 1}^m s_l$, $R_l = \sum_{l = 1}^m R_l$. Taking summation over $l$ and replacing $M_1$ by $c_1c_*/c^*$ completes the whole proof.
\end{proof}
\end{document}